\documentclass[a4paper,UKenglish,runningheads]{llncs}
\usepackage{microtype}

\usepackage{amsmath,amssymb}
\usepackage{mathrsfs}
\usepackage{style/mathpartir}
\usepackage[numbers]{natbib}
\usepackage{csquotes}

\input{style/layout}

\usepackage{centernot}
\usepackage[Symbol]{upgreek}
\usepackage{stmaryrd}

\usepackage{mathtools}

\let\originalleft\left
\let\originalright\right
\renewcommand{\left}{\mathopen{}\mathclose\bgroup\originalleft}
\renewcommand{\right}{\aftergroup\egroup\originalright}

\makeatletter

\let\ams@cases@copy@\cases
\def\cases{%
  \newcommand{\CASE}{&\text{if }}%
  \newcommand{\AND}{\\&\text{and }}%
  \newcommand{\OTHERWISE}{& \text{otherwise}}%
  \ams@cases@copy@%
}

\let\stdphi\phi
\let\phi\varphi
\let\varphi\stdphi

\def\pre#1-{$#1$\nobreak-\nobreak\hskip0pt}

\newcommand{\picalc}{%
    \texorpdfstring{\pre\pi-calculus}{pi-calculus}}

\newcommand{\piterm}{%
    \texorpdfstring{\pre\pi-term}{pi-term}}

\def\CCS#1{\texorpdfstring{\textsc{CCS}$^{#1}$}{CCS#1}}

\newcommand{\PiTerms}{\mathcal{P}}
\newcommand{\PiNf}{\PiTerms_{\!\mathsf{nf}}}
\newcommand{\PiAnnot}[1][\HTypes]{\PiTerms^{#1}}
\newcommand{\PiNfAnnot}[1][\HTypes]{\PiNf^{#1}}
\newcommand{\Seq}{\mathcal{S}}
\newcommand{\DBTerms}{\mathcal{P}_{\!\mathsf{db}}}

\newcommand{\Forests}{\mathcal{F}}

\def\preTypes{\pre{\Types\kern-.5ex}}

\newcommand{\tcompat}[1][ible]{%
  \texorpdfstring{\preTypes-compat#1}{T-compat#1}}
\def\tcompatibility{\texorpdfstring{\preTypes-compatibility}{T-compatibility}}
\newcommand{\tshaped}[1][]{%
  \texorpdfstring{\preTypes-shaped#1}{T-shaped#1}}

\let\arrowvect\vec

\newcommand{\vect}[1]{\arrowvect{#1}\@ifnextchar]{\,}{\@ifnextchar){\,}{\@ifnextchar{\rangle}{\,}{}}}}
\renewcommand{\vec}[1]{\mathbf{#1}}

\newcommand{\sem}[1]{\llbracket#1\rrbracket}

\newcommand{\Nat}{\mathbb{N}}

\newcommand{\union}{\cup}
\newcommand{\inters}{\cap}
\newcommand{\Union}{\bigcup}

\newcommand{\dunion}{\uplus}
\newcommand{\Dunion}{\biguplus}
\DeclarePairedDelimiter\card{\lvert}{\rvert}

\providecommand{\implies}{\Rightarrow}

\newcommand{\domain}{\operatorname{dom}}

\newcommand{\lst}[3][1]{{#2_{#1}}\ldots {#2_{#3}}}    %
\newcommand{\lstc}[3][1]{{#2_{#1}},\ldots, {#2_{#3}}} %

\providecommand{\coloneq}{\mathrel{\mathop:}=} 
\providecommand{\Coloneqq}{\mathrel{\mathop{::}}=} 
\newcommand{\is}{\coloneq}

\def\congr{\equiv}

\newcommand{\reach}{\operatorname{Reach}}

\newcommand{\from}{\colon}

\newcommand{\st}{.\:}  

\let\SavedDoubleVert\relax

{\catcode`\|=\active
  \xdef\set{\protect\expandafter\noexpand\csname set \endcsname}
  \expandafter\gdef\csname set \endcsname#1{\begingroup\mathinner%
  \ifx!#1!%
      \emptyset%
  \else%
      {\lbrace\mathcode`\|32768\let|\midvert #1\rbrace}%
  \fi%
  \endgroup%
  }
  \xdef\Set{\protect\expandafter\noexpand\csname Set \endcsname}
  \expandafter\gdef\csname Set \endcsname#1{%
      \ifx!#1!%
         \emptyset%
      \else%
         \left\{%
         \ifx\SavedDoubleVert\relax \let\SavedDoubleVert\|\fi
         \,{\let\|\SetDoubleVert
         \mathcode`\|32768\let|\SetVert
         #1}\,\right\}%
      \fi%
  }
}
\def\midvert{\egroup\mid\bgroup}
\def\SetVert{\@ifnextchar|{\|\@gobble}
    {\egroup\;\mid@vertical\;\bgroup}}
\def\SetDoubleVert{\egroup\;\mid@dblvertical\;\bgroup}

\begingroup
 \edef\@tempa{\meaning\middle}
 \edef\@tempb{\string\middle}
\expandafter \endgroup \ifx\@tempa\@tempb
 \def\mid@vertical{\middle|}
 \def\mid@dblvertical{\middle\SavedDoubleVert}
\else
 \def\mid@vertical{\mskip1mu\vrule\mskip1mu}
 \def\mid@dblvertical{\mskip1mu\vrule\mskip2.5mu\vrule\mskip1mu}
\fi

\newcommand{\map}[1]{%
    \if\relax\noexpand#1\relax%
        \emptyset%
    \else%
        [\,\@map#1;\relax\noexpand\@end@map\,]%
    \fi%
}
\newcommand{\Map}[1]{%
    \if\relax\noexpand#1\relax%
        \emptyset%
    \else%
        \left[\,\@map#1;\relax\noexpand\@end@map\,\right]%
    \fi%
}
\def\@map#1;#2\@end@map{
    \ifx\relax#2\relax%
        \map@binding[#1]
    \else%
        \map@binding[#1],\:\@map#2\@end@map
    \fi%
}
\def\map@binding[#1->#2]{#1\mapsto#2}

\newcommand{\subst}[1]{%
    \if\relax\noexpand#1\relax%
    \else%
        [\,\@subst#1,\relax\noexpand\@end@subst\,]%
    \fi%
}
\def\@subst#1,#2\@end@subst{
    \ifx\relax#2\relax%
        \subst@binding[#1]
    \else%
        \subst@binding[#1],\:\@subst#2\@end@subst
    \fi%
}

\def\subst@binding[#1->#2]{#2/#1}

\def\grammOr{\hspace{3pt}\mid\hspace{3pt}}
\def\grammIs{\Coloneqq}

\begingroup
\catcode`\|=\active%
\gdef\@grammar@bar{%
    \catcode`\|=\active%
    \def|{\grammOr}%
}
\endgroup

\newcommand{\gramm}[1]{%
  \begingroup
  \def\is{\grammIs}%
  \@grammar@bar%
  #1%
  \endgroup%
}

\newenvironment{grammar}{%
    \begin{equation*}%
    \def\is{& \grammIs }%
    \@grammar@bar%
    \aligned%
}
{%
    \endaligned%
    \end{equation*}%
    \aftergroup\ignorespaces%
}

\newcommand{\Names}{\ensuremath{\mathcal{N}}}

\def\out#1<#2>{\overline{#1}\langle#2\rangle}
\def\inp#1(#2){{#1}(#2)}
\def\tact{\boldsymbol\tau} 
\let\inpz\relax
\def\outz#1{\overline{#1}}
\def\new#1.{\restr #1.\ignorespaces}

\newcommand{\restr}{\upnu}  
\newcommand{\zero}{\mathbf{0}}
\newcommand{\freenames}{\operatorname{fn}}
\newcommand{\boundnames}{\operatorname{bn}}
\newcommand{\resboundnames}{\operatorname{bn}_\nu}

\newcommand{\actrestr}{\operatorname{act}_{\restr}}
\newcommand{\seqproc}{\operatorname{act}_{\Seq}}

\newcommand{\nf}{\operatorname{nf}}

\newcommand{\nestr}{\operatorname{nest}_{\restr}}
\newcommand{\depth}{\operatorname{depth}}

\newcommand{\ch}[1]{\mathit{#1}}

\newcommand{\emptyforest}{(\emptyset, \emptyset)}
\newcommand{\forest}{\operatorname{forest}}
\newcommand{\AST}{\operatorname{\mathcal{F}}\sem}
\newcommand{\height}{\operatorname{height}}

\newcommand{\CommTop}{\operatorname{\mathcal{G}}\sem}

\newcommand{\embeddedin}{\preceq}

\newcommand{\minrestr}{\operatorname{min}_{\Types}}

\newif\if@initial@pi@term@

\newcommand{\redto}{\to}

\newcommand{\bang}[1]{{!}#1}
\newcommand{\Bang}[1]{\bang{\left(\vphantom{\big(}#1\right)}}

\newcommand{\Parallel}{\@ifstar{\prod}{{\textstyle\prod}}}
\newcommand{\Alt}{\@ifstar{\sum}{{\textstyle\sum}}}

\newcommand{\linkedto}[1]{\leftrightarrow_{#1}}
\newcommand{\tiedto}[1]{\smallfrown_{#1}}
\newcommand{\ntiedto}[1]{\triangleleft_{#1}}
\newcommand{\migr}[1]{\operatorname{Mig}_{#1}}

\def\pntr[#1>{\mathop{[#1\rangle}}
\def\pnstep{\mathop{[\mkern-1mu\rangle}}

\newcommand{\Types}{\ensuremath{\mathcal{T}}} 
\newcommand{\HTypes}{\ensuremath{\mathbb{T}}} 

\newcommand{\Env}{\Gamma}
\newcommand{\types}[1][\Types]{\vdash_{#1}}
\newcommand{\type}{\tau} 
\newcommand{\tas}{\,{:}\,} 
\newcommand{\ty}[1]{\mathsf{#1}} 

\newcommand{\tlt}{<}

\newcommand{\parent}{%
  \@ifnextchar+{\p@rentop@plus}{%
  \@ifnextchar*{\p@rentop@star}{%
  \p@rentop}}}
\stmry@if\DeclareMathSymbol\YleftRel\mathrel{stmry}{"06}\fi
\stmry@if\DeclareMathSymbol\YrightRel\mathrel{stmry}{"07}\fi
\newcommand{\rparent}{\YrightRel}

\newcommand{\p@rentop}{\YleftRel}
\newcommand{\p@rentop@plus}[1]{<}
\newcommand{\p@rentop@star}[1]{\leq}

\newcommand\erase[1]{{\ulcorner}{#1}{\urcorner}}

\newcommand{\btyof}[1]{t_{\ch{#1}}}

\newcommand{\typevar}{\mathfrak{t}}  
\newcommand{\constraints}{\mathcal{C}}

\newcommand{\Premise}{\Psi}

\newcommand{\phimig}{\phi_{\mathrm{mig}}}
\newcommand{\phinonmig}{\phi_{\neg\mathrm{mig}}}

\newcommand{\treeins}{\operatorname{ins}}

\newcommand{\base}{\operatorname{base}}

\makeatother

\spnewtheorem*{nameuniq}{Name Uniqueness Assumption}{\bfseries}{\itshape}
\makeatletter
\let\@nameuniq\nameuniq
\def\nameuniq{%
  \@nameuniq%
  \def\@currentlabel{Name Uniqueness}
  \phantomsection
  \label{nameuniq}
}
\makeatother

\newenvironment{proof*}
    {\begin{proof}[Proof \textup(Sketch\textup)]}
    {\end{proof}}

\usepackage{tikz}
\usetikzlibrary{calc,decorations.pathreplacing,shapes.misc,scopes,positioning}
\usetikzlibrary{external}

\newcommand{\ExportTikzPictures}{\tikzexternalize[prefix=export/]}

\newcommand{\inputfigwrap}[2][]{\tikzsetnextfilename{#2}\begin{tikzpicture}[#1]\input{figures/#2.tikz}\end{tikzpicture}}

\pgfdeclarelayer{background}
\pgfsetlayers{background,main}

\makeatletter
\newcommand{\savePictureHeight}[1]{%
  \pgfpointdiff{\pgfpointanchor{current bounding box}{north}}{\pgfpointanchor{current bounding box}{south}}%
  \pgfmathparse{veclen(\pgf@x,\pgf@y)}%
  \global\let#1\pgfmathresult\relax%
}
\newcommand{\matchPictureHeight}[1]{%
  \savePictureHeight{\pic@height@tmp}
  \pgfmathsetmacro{\pic@height@tmp}{(\pic@height@tmp - #1)/2}
  \useasboundingbox
    (current bounding box.north) ++(0,-\pic@height@tmp pt)
    (current bounding box.south) ++(0,\pic@height@tmp pt);
}
\makeatother

\tikzset{
  border/.style={
    draw=#1!50,fill=#1!5,rounded corners
  },
  border/.default=black,
  bordered/.style={
    execute at end picture={
        \begin{pgfonlayer}{background}
            \path[border=#1] ($(current bounding box.north west)+(-2mm,3mm)$)
                   rectangle ($(current bounding box.south east)+(2mm,-3mm)$);
        \end{pgfonlayer}
    }
  },
  bordered/.default=black,
}

\tikzset{
    show control points/.style={
        decoration={
            show path construction,
            curveto code={
                \draw [blue, dashed]
                    (\tikzinputsegmentfirst) -- (\tikzinputsegmentsupporta)
                    node [at end, cross out, draw, solid, red, inner sep=2pt]{};
                \draw [blue, dashed]
                    (\tikzinputsegmentsupportb) -- (\tikzinputsegmentlast)
                    node [at start, cross out, draw, solid, red, inner sep=2pt]{};
            }
        },
        postaction=decorate
    },
}
\colorlet{treecol}{white}
\colorlet{lightgray}{black!10}

\tikzset{
  mini tree/.style = {
    AST,level distance=3.5ex,sibling distance=3.5ex,baseline=-2ex
  },
  forest/.style={
    node/.style = {
      label=##1,
      circle,
      draw=black,
      fill=black,
      inner sep=0,
      outer sep=\pgflinewidth,
      minimum size=3pt
    },
    node/.default = {},
    removed node/.style = {
      cross out,
      thick,
      draw=black,
      fill=none,
      inner sep=0,
      outer sep=0,
      label={[gray]##1},
      minimum size=5pt
    },
    removed node/.default = {},
    closer/.style={label distance=-.4em},
    tree/.style={fill=##1!90!black, draw=##1!70!black},
    tree/.default = treecol,
    migrating/.style={tree=red},
    receiver cont/.style={tree=blue,transform shape},
    sender cont/.style={tree=green!90!gray,transform shape},
    subtree/.pic={
      \draw[line join=round,pic actions] (0,0) coordinate (-root) -- ++(-.2,-.5) --coordinate[midway] (-bottom) ++(.4,0) --cycle;
    },
    continuation subtree/.style={
      pic type=subtree,
      tree=lightgray,
      transform shape
    },
    migrating subtree/.style={
      pic type=subtree,
      migrating,
      transform shape
    },
    sender subtree/.style={
      pic type=subtree,
      sender cont,
      transform shape
    },
    receiver subtree/.style={
      pic type=subtree,
      receiver cont,
      transform shape
    },
  },
  AST/.style = {
    child anchor=north,
    every node/.style={
      outer sep=0pt,
      inner sep=2pt,
    },
    level distance=2em,
    sibling distance=8mm,
    many children/.style = {
      edge from parent path = {
      [gray]
        (\tikzparentnode)
          edge (\tikzchildnode.north east)
          edge (\tikzchildnode.north)
          edge (\tikzchildnode.north west)
      }
    },
    dots/.style = {
      inner sep = -2pt,
      node contents = {$\strut\cdots$}
    },
  },
  forests distance/.store in=\forestsdist,
  forests distance=3,
}

\colorlet{linkcolor}{black}
\definecolor{seqcolor1}{RGB}{174, 199, 232}
\colorlet{seqcolor}{seqcolor1!30}
\colorlet{spawncol}{blue!60!black}
\colorlet{ghostcol}{black!30}

\tikzset{
  stdnf graph/.style = {
    -,draw,semithick,
    ghost/.style={draw=ghostcol,ch label/.style={color=white}},
    seq/.style={
        rectangle,
        rounded corners=.2em,
        draw=black,
        fill=seqcolor,
        minimum size=1em,
        inner sep=3pt,
        ghost/.style={fill=seqcolor!30,draw=ghostcol},
    },
    link/.style={
      draw=linkcolor
    },
    every edge/.style=link,
    ch/.style={
        circle,
        fill=linkcolor,
        draw=white, 
        minimum size=1ex,
        outer sep=0,
        inner sep=0,
        label={[ch label]##1},
        ghost/.style={fill=ghostcol},
    },
    ch/.default={},
    ch label/.style={
        font={\scriptsize},
        outer sep = 0,
        inner sep = .3ex,
        rectangle
    },
    bang/.style={
        dashed,thick,rounded corners,
        draw=black!50!blue,
        inner sep=5pt,
        append after command={(\tikzlastnode.north east) node {\Huge \textcolor{black!50!blue}{*}}},
        fit=##1
    },
    tiny/.style={
        minimum size=3pt,
        inner sep=0pt,
    },
    bigger nodes/.style={
        seq/.append style={
            minimum size=13pt,
        },
        ch/.append style={
            minimum size=10pt,
        },
    }
  }
}

\makeatletter
\newcommand{\ensuretikzpicturebegin}{%
  \ifx\pgfpictureid\@undefined%
    \def\ensuretikzpictureend{\end{tikzpicture}}%
    \begin{tikzpicture}%
  \else%
    \def\ensuretikzpictureend{}%
  \fi%
}
\makeatother
\ExportTikzPictures 

\newif\ifshortversion
\def\iflongversion{\ifshortversion\else} 

\newif\ifincludeappendix
\includeappendixtrue

\def\appendixorfull{%
  \ifincludeappendix%
    \hyperref[appendix]{Appendix}%
  \else%
    \cite{fullversion}%
  \fi%
}

\AtBeginDocument{
  \ifincludeappendix%
    \setcounter{tocdepth}{1}%
  \fi%
}

\usepackage[capitalise]{cleveref}

\title{On Hierarchical Communication Topologies in the \picalc}

\author{Emanuele D'Osualdo\inst{1} \and C.-H. Luke Ong\inst{2}}
\institute{%
  TU Kaiserslautern
  \email{dosualdo@cs.uni-kl.de}%
\and
  University of Oxford
  \email{lo@cs.ox.ac.uk}%
}
\authorrunning{E. D'Osualdo and C.-H.~L. Ong}

\begin{document}

\maketitle

\begin{abstract}
This paper is concerned with the shape invariants satisfied by the communication topology of \piterm{s}, and the automatic inference of these invariants.
A \piterm\ $P$ is \emph{hierarchical} if there is a finite forest \Types\ such that the communication topology of every term reachable from $P$ satisfies a \tshaped\ invariant.
We design a static analysis to prove a term hierarchical by means of a novel type system that enjoys decidable inference.
The soundness proof of the type system employs a non-standard view of \picalc{} reactions.
The coverability problem for hierarchical terms is decidable.
This is proved by showing that every hierarchical term is depth-bounded,
an undecidable property known in the literature.
We thus obtain an expressive static fragment of the \picalc{} with decidable safety verification problems.
\end{abstract}

\section{Introduction}
\label{sec:intro}

Concurrency is pervasive in computing.
A standard approach is to organise concurrent software systems as a dynamic collection of processes that communicate by message passing.
Because processes may be destroyed or created, the number of processes in the system changes in the course of the computation, and may be unbounded.
Moreover the messages that are exchanged may contain process addresses.
Consequently the \emph{communication topology} of the system%
---the hypergraph \cite{Milner:92,Milner:99} connecting processes that can communicate directly---%
evolves over time.
In particular, the connectivity of a process
(i.e.~its neighbourhood in this hypergraph)
can change dynamically.
The design and analysis of these systems is difficult:
the dynamic reconfigurability alone renders verification problems undecidable.
This paper is concerned with \emph{hierarchical systems},
a new subclass of concurrent message-passing systems
that enjoys decidability of safety verification problems,
thanks to a shape constraint on the communication topology.

The \picalc\ of Milner, Parrow and Walker \cite{Milner:92} is a process calculus designed to model systems with a dynamic communication topology.
In the \picalc{}, processes can be spawned dynamically, and they communicate by exchanging messages along synchronous channels.
Furthermore channel names can themselves be created dynamically, and passed as messages, a salient feature known as \emph{mobility},
as this enables processes to modify their neighbourhood at runtime.%

It is well known that the \picalc{} is a Turing-complete model of computation.
Verification problems on \piterm{s} are therefore undecidable in general.
There are however useful fragments of the calculus
that support automatic verification.
The most expressive such fragment known to date is the \emph{depth-bounded} \picalc\ of Meyer~\cite{Meyer:08}.
Depth boundedness is a constraint on the shape of communication topologies.
A \piterm\ is \emph{depth-bounded} if there is a number $k$
such that every simple path%
\footnote{a simple path is a path with no repeating edges.}
in the communication topology
of every reachable \piterm\ has length bounded by $k$.
Meyer~\cite{Meyer:phd} proved that termination and coverability
(a class of safety properties)
are decidable for depth-bounded terms.

Unfortunately depth boundedness itself is an undecidable property~\cite{Meyer:phd}, which is a serious impediment to the practical application of the depth-bounded fragment to verification.
This paper offers a two-step approach to this problem.
First we identify a (still undecidable) subclass of depth-bounded systems, called \emph{hierarchical}, by a shape constraint on
communication topologies (as opposed to numeric, as in the case of depth-boundedness).
Secondly, by exploiting this richer structure, we define a type system, which in turn gives a \emph{static} characterisation of an expressive and practically relevant fragment of the depth-bounded \picalc{}.

\begin{example}[Client-server pattern]
\label{ex:cs-pattern}
To illustrate our approach,
consider a simple system implementing a client-server pattern.
A server $S$ is a process listening on a channel $s$ which acts as its address.
A client $C$ knows the address of a server and has a private channel $c$ that represents its identity.
When the client wants to communicate with the server,
it asynchronously sends 
$c$ along the channel~$s$.
Upon receipt of the message,
the server acquires knowledge of (the address of) the requesting client;
and spawns a process $A$ to answer the client's request $R$ asynchronously;
the answer consists of a new piece of data,
represented by a new name $d$, sent along the channel $c$.
Then the server forgets the identity of the client and reverts to listening for new requests.
Since only the requesting client knows $c$ at this point,
the server's answer can only be received by the correct client.
\Cref{fig:server-sketch:protocol} shows the communication topology
of a server and a client, in the three phases of the protocol.

\begin{figure}[tb]
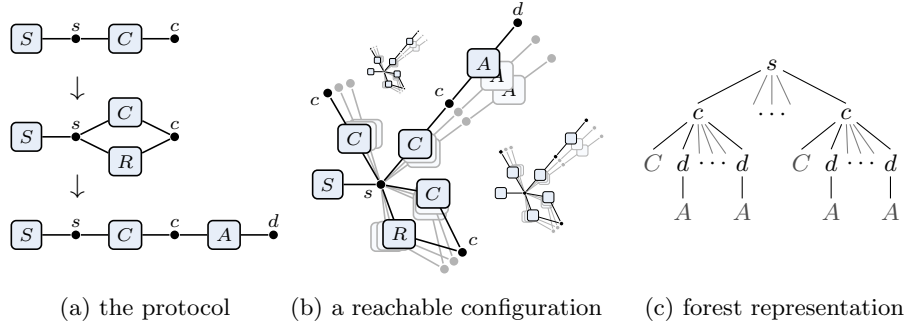

  \centering%
  \subfloat[the protocol]{%
    \input{export/server-sketch-protocol.patch.tex}%
    \label{fig:server-sketch:protocol}%
  }%
  \subfloat[a reachable configuration]{%
    \input{export/server-sketch-topology.patch.tex}%
    \label{fig:server-sketch:topology}%
  }\quad
  \subfloat[forest representation]{%
    \input{export/server-sketch-forest.patch.tex}%
    \label{fig:server-sketch:forest}%
  }%
  \caption{%
    Evolution of the communication topology
    of a server interacting with a client.
    $R$ represents a client's pending request and
    $A$ a server's pending answer.
  }%
  \label{fig:server-sketch}%
\end{figure}

The overall system is composed of an unbounded number of servers and clients, constructed according to the above protocol.
The topology of a reachable configuration is depicted
in \cref{fig:server-sketch:topology}.
While in general the topology of a mobile system can become arbitrarily complex,
for such common patterns as client-server,
the programmer often has a clear idea of the desired shape of the communication topology:
there will be a number of servers, each with its cluster of clients;
each client may in turn be waiting to receive a number of private replies.
This suggests a hierarchical relationship between the names
representing servers, clients and data,
although the communication topology itself does not form a tree.
\end{example}%

\subsubsection*{\tcompat[ibility] and hierarchical terms} %

Recall that in the \picalc{} there is an important relation between terms, $\congr$, called \emph{structural congruence},
which equates terms that differ only in irrelevant presentation details,
but not in behaviour.
For instance, the structural congruence laws for restriction tell us
that the order of restrictions is irrelevant---%
$\new x.\new y.P \congr \new y.\new x.P$---%
and that the scope of a restriction can be extended
to processes that do not refer to the restricted name---%
i.e.,~$(\new x.P) \parallel Q \congr \new x.(P \parallel Q)$
when $x$ does not occur free in $Q$---%
without altering the meaning of the term.
The former law is called \emph{exchange},
the latter is called \emph{scope extrusion}.

Our first contribution is a formalisation in the \picalc\ of the intuitive notion of hierarchy illustrated in Example~\ref{ex:cs-pattern}.
We shall often speak of the \emph{forest representation} of a \piterm\ $P$, $\forest(P)$,
which is a version of the abstract syntax tree of $P$ that captures the nesting relationship between the active restrictions of the term.
(A restriction of a \piterm\ is \emph{active} if it is not in the scope of a prefix.)
Thus the internal nodes of a forest representation are labelled with (active) restriction names, and its leaf nodes are labelled with the sequential subterms.
Given a \piterm\ $P$, we are interested in not just $\forest(P)$, but also $\forest(P')$ where $P'$ ranges over the structural congruents of $P$, because these are all behaviourally equivalent representations.
See Fig.~\ref{fig:forests} for an example of the respective forest representations of the structural congruents of a term.
In our setting a hierarchy $\Types$ is a \emph{finite} forest
of what we call \emph{base types}.
Given a finite forest $\Types$, we say that a term $P$ is \emph{\tcompat} if there is a term $P'$, which is structurally congruent to $P$, such that the parent relation of $\forest(P')$ is consistent with the partial order of $\Types$.

In Example~\ref{ex:cs-pattern} we would introduce base types
  $\ty{srv}$, $\ty{cl}$ and $\ty{data}$
associated with the restrictions
  $\restr s$, $\restr c$ and $\restr d$
respectively,
and we would like the system to be compatible to the hierarchy
$\Types = \ty{srv} \parent \ty{cl} \parent \ty{data}$,
where $\parent$ is the is-parent-of relation.
That is,
we must be able to represent a configuration
with a forest that, for instance,
does not place a server name below a client name nor
a client name below another client name.
Such a representation is shown in \cref{fig:server-sketch:forest}.

In the Example, we want every reachable configuration of the system to be compatible
with the hierarchy.
We say that a \piterm\ $P$ is \emph{hierarchical} if there is a hierarchy $\Types$ such that every term reachable from $P$ is \tcompat.
Thus the hierarchy $\Types$ is a shape invariant of the communication topology under reduction.

It is instructive to express depth boundedness as a constraint on forest representation:
a term $P$ is depth-bounded if there is a constant $k$ such that every term reachable from $P$ has a structurally congruent $P'$ whereby $\forest(P')$ has height bounded by $k$.
It is straightforward to see that hierarchical terms are depth-bounded; the converse is however not true.%

\subsubsection*{A type system for hierarchical terms} %

While membership of the hierarchical fragment is undecidable,
by exploiting the forest structure,
we have devised a novel type system that guarantees the invariance of \tcompat[ibility] under reduction.
Furthermore type inference is decidable, so that the type system can be used to
infer a hierarchy \Types\, with respect to which the input term is hierarchical.
To the best of our knowledge, our type system is the first
that can infer a shape invariant of the communication topology of a system.

The typing rules that ensure invariance of \tcompat[ibility] under reduction arise from a new perspective of the \picalc{} reaction,
one that allows compatibility to a given hierarchy to be tracked more readily.
Suppose we are presented with a \tcompat\ term $P = C[S, R]$ where $C[\hbox{-},\hbox{-}]$ is the \emph{reaction context}, and the two processes $S = \out a<b>.S'$ and $R = \inp a(x).R'$ are ready to communicate over a channel $a$.
After sending the message $b$, $S$ continues as the process $S'$, while
upon receipt of $b$, $R$ binds $x$ to $b$ and continues as $R'' = R'\subst{x -> b}$.
Schematically, the traditional understanding of this transaction is:
first extrude the scope of $b$ to include $R$,
then let them react, as shown in \cref{fig:trad-react}.

\begin{figure}[tb]
  \centering%
  \input{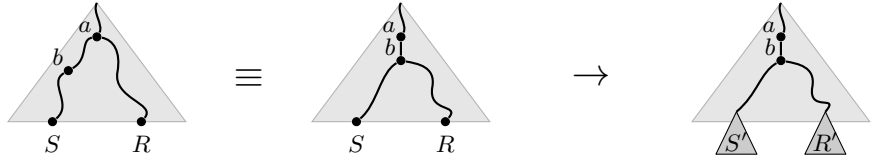}%
  \caption{Standard view of \picalc{} reactions}%
  \label{fig:trad-react}
\end{figure}

Instead, we seek to implement the reaction \emph{without scope extrusion}:
after the message is transmitted, the sender continues in-place as $S'$,
while $R''$ is split in two parts
$R_{\text{mig}}' \parallel R_{\neg\text{mig}}'$,
 one that uses the message
(the \emph{migratable} part) and one that does not.
As shown in \cref{fig:tcompat-react},
the migratable part of $R''$, $R_{\text{mig}}'$, is \enquote{installed} under $b$ so that it can
make use of the acquired name,
while the non-migratable one, $R_{\neg\text{mig}}'$, can simply continue in-place.

\begin{figure}[tb]
  \centering%
  \input{export/sync-tcomp.patch.tex}%
  \caption{\protect{\tcompatibility} preserving reaction}%
  \label{fig:tcompat-react}%
\end{figure}

Crucially, the \emph{reaction context}, $C[\hbox{-}, \hbox{-}]$, is left unchanged.
This means that if the starting term is \tcompat,
the reaction context of the \emph{reactum} is \tcompat\ as well.
We can then focus on imposing constraints on the use of names of $R'$ so that the migration does not result in $R_{\text{mig}}'$ escaping the scope of previously bound names.

By using these ideas, our type system is able to statically accept \picalc\ encodings of such system as that discussed in Example~\ref{ex:cs-pattern}.
The type system can be used, not just to \emph{check} that a given $\Types$ is respected by the behaviour of a term, but also to \emph{infer} a suitable $\Types$ when it exists.
Once typability of a term is established,
safety properties such as unreachability of error states, mutual exclusion or bounds on mailboxes, can be verified algorithmically.
For instance, in Example~\ref{ex:cs-pattern}, a coverability check can prove that each client can have at most one reply pending in its mailbox.
To prove such a property, one needs to construct an argument that reasons about dynamically created names with a high degree of precision.
This is something that counter abstraction and uniform abstractions based methods
have great difficulty attaining.

Our type system is (necessarily) incomplete in that there are depth-bounded,
or even hierarchical,
systems that cannot be typed.
The class of \piterm{s} that can be typed is non-trivial, and includes terms which generate an unbounded number of names and exhibit mobility.%

\paragraph{Outline.}
In \cref{sec:picalc} we review the \picalc{}, depth-bounded terms, and related technical preliminaries.
In \cref{sec:t-compat} we introduce \tcompat[ibility] and the hierarchical terms.
We present our type system in \cref{sec:typesys}.
\Cref{sec:soundness} discusses soundness of the type system.
In \cref{sec:inference} we give a type inference algorithm; and in \cref{sec:verification} we  present results on expressivity and discuss applications.
We conclude with related and future work in \cref{sec:relwork,sec:future}.
All missing definitions and proofs can be found in \appendixorfull.

\section{The \picalc\ and the depth-bounded fragment}
\label{sec:picalc}\label{sec:prelim}

\subsection{Syntax and semantics}
We use a \picalc{} with guarded replication to express recursion~\cite{Milner:92a}.
Fix a universe $\Names$ of \emph{names} representing channels and messages.
The syntax is defined by the grammar:
\begin{grammar}
    \PiTerms \ni P, Q \is \zero | \new x.P | P_1 \parallel P_2 | M | \bang M
        && \text{process}\\
    M \is M + M | \pi. P
        && \text{choice}\\
    \pi \is \inp a(x) | \out a<b> | \tact
        && \text{prefix}
\end{grammar}
\begin{definition}\label{def:congr}
Structural congruence, $\congr$,
is the least relation that respects $\alpha$-conversion of bound names,
and is associative and commutative with respect to
  $+$ (choice) and $\parallel$ (parallel composition)
  with $\zero$ as the neutral element,
and satisfies laws for restriction:
  $ \new a . \zero \congr \zero $ and
  $ \new a . \new b . P \congr \new b . \new a . P $,
and
\begin{align*}
  \bang P & \congr P \parallel \bang P
  & \hbox{\emph{Replication}}
  \\
  P \parallel \new a . Q
  & \congr
  \new a . (P \parallel Q)
  \quad
  (\hbox{if $a \not\in \freenames(P)$})
  \qquad
  &
  \hbox{\emph{Scope Extrusion}}
\end{align*}%
\end{definition}
In $P = \pi.Q$, we call $Q$ the \emph{continuation} of $P$ and
will often omit $Q$ altogether when $Q = \zero$.
In a term $\new x.P$ we will occasionally
refer to $P$ as the \emph{scope} of $x$.
The name $x$ is bound in both $\new x.P$, and in $\inp a(x).P$.
We will write $\freenames(P)$, $\boundnames(P)$ and $\resboundnames(P)$
for the set of free, bound and restriction-bound names in $P$, respectively.
A sub-term is \emph{active} if it is not under a prefix.
A name is active when it is bound by an active restriction.
We write $\actrestr(P)$ for the set of active names of $P$.
Terms of the form $M$ and $\bang M$ are called \emph{sequential}.
We write $\Seq$ for the set of 
sequential terms, 
$\seqproc(P)$ for the set of 
active sequential processes of~$P$,
and $P^i$ for the parallel composition
of $i$ copies of $P$.

Intuitively, a sequential process acts like a thread
running finite-control sequential code.
A term 
$\tact.(P\parallel Q)$ is the equivalent of spawning a process $Q$
and continuing as $P$%
---although in this context the r\^oles of $P$ and $Q$ are interchangeable.
Interaction is by \emph{synchronous} communication over channels.
An \emph{input prefix} $\inp a(x)$ is a blocking receive on the channel $a$
binding the variable $x$ to the message.
An \emph{output prefix} $\out a<b>$ is a blocking send of the message $b$
along the channel $a$;
here $b$ is itself the name of a channel
that can be used subsequently for further communication:
an essential feature for mobility.
A non-blocking send can be simulated
by spawning a new process doing a blocking send.
Restrictions are used to make a channel name private.
A replication $\bang{(\pi.P)}$ can be understood as having a server that can spawn a new copy of $P$ whenever a process tries to communicate with it.
In other words it behaves like an infinite parallel composition
$(\pi.P \parallel \pi.P \parallel \cdots)$.

For conciseness, we assume channels are unary (the extension to the polyadic case is straightforward).
In contrast to process calculi without mobility, replication and systems of tail recursive equations are equivalent methods of defining recursive processes in the \picalc{} \cite[Section 3.1]{Milner:93}.



We rely on the following mild assumption,
that the choice of names is unambiguous,
especially when selecting a representative for a congruence class:

\begin{nameuniq}
  Each name in $P$ is bound at most once and
  $\freenames(P) \inters \boundnames(P) = \emptyset$.
\end{nameuniq}

\paragraph{Normal Form.}
The notion of hierarchy, which is central to this paper,
and the associated type system depend heavily on structural congruence.
These are criteria that, given a structure on names,
require the existence of a specific representative
of the structural congruence class
exhibiting certain properties.
However, we cannot assume the input term is presented as that representative;
even worse, when the structure on names is not fixed
(for example, when inferring types)
we cannot fix a representative and be sure that
it will witness the desired properties.
Thus, in both the semantics and the type system,
we manipulate a neutral type of representative called \emph{normal form},
which is a variant of the \emph{standard form}~\cite{Milner:92}.
In this way we are not distracted by
the particular syntactic representation we are presented with.

We say that a term $P$ is in \emph{normal form} ($P \in \PiNf$) if it is in standard form
and each of its inactive subterms is also in normal form.
Formally, normal forms are defined by the grammar
%
\begin{grammar}
    \PiNf \ni N \is \new x_1.\cdots \new x_n.
                      (A_1 \parallel \cdots \parallel A_m) \\
    A \is \pi_1. N_1 + \cdots + \pi_n. N_n |  \bang{\left(\pi_1. N_1 + \cdots + \pi_n. N_n \right)} \\
\end{grammar}
where the sequences $\lst{x}{n}$ and $\lst{A}{m}$ may be empty;
when they are both empty the normal form is the term $\zero$.
We further assume w.l.o.g. that 
normal forms satisfy \ref{nameuniq}.
Given a finite set of indexes $I = \set{i_1,\dots,i_n}$ we write
$\Parallel_{i \in I} A_i$ for $(A_{i_1} \parallel \cdots \parallel A_{i_n})$,
which is $\zero$ when $I$ is empty;
and $\Alt_{i \in I} \pi_i. N_i$ for
$(\pi_{i_1}. N_{i_1} + \cdots + \pi_{i_n}. N_{i_n})$.
This notation is justified by commutativity and associativity
of the parallel and choice operators.
Thanks to the structural laws of restriction,
we also write $\new X.P$ where $X = \set{\lstc{x}{n}}$,
or $\new x_1\:x_2\cdots x_n.P$, for $\new x_1.\cdots \new x_n.P$;
or just $P$ when $X$ is empty.
When $X$ and $Y$ are disjoint sets of names, we use juxtaposition for union.

Every process $P \in \PiTerms$ is structurally congruent
to a process in normal form.
The function $\nf \from \PiTerms \to \PiNf$,
defined in \appendixorfull,
extracts, from a term, a structurally congruent normal form.

Given a process $P$ with normal form  $\new X.\Parallel_{i \in I} A_i$,
the \emph{communication topology}%
\footnote{This definition arises from the ``flow graphs'' of~\cite{Milner:92}; see e.g.~\cite[p.~175]{Meyer:phd} for a formal definition.}
of $P$, written
$\CommTop{P}$, is defined as the labelled hypergraph
with
  $X$ as hyperedges and
  $I$ as nodes, each labelled with the corresponding $A_i$.
  An hyperedge $x \in X$ is connected with $i$
  just if $x \in \freenames(A_i)$.

\paragraph{Semantics.}
\label{sec:semantics}
We are interested in the reduction semantics of a \piterm{},
which can be described using the following rule.

\begin{definition}[Semantics of \picalc]
  \label{def:sem-of-picalc}
  The operational semantics of a term $P_0 \in \PiTerms$ is defined
  by the (pointed) transition system $(\PiTerms, \redto, P_0)$ on \piterm{s},
  where $P_0$ is the initial term,
  and the transition relation,
    ${\redto} \subseteq \PiTerms^2$,
  is defined by $P\redto Q$ if either
  \ref{redex-P} to \ref{redex-Q} hold,
  or
  \ref{redex-tau-P} and \ref{redex-tau-Q}
  hold, where \\[1ex]
  \begin{minipage}{.5\linewidth}
  \begin{defenum}[series=semantics]
    \item
      $P \congr \new W.(S \parallel R \parallel C) \in \PiNf$,
      \label{redex-P}
    \item
      $S = (\out a<b>.\new Y_s.S')+M_s$,
      \label{redex-S}
    \item
      $R = (\inp a(x).\new Y_r.R')+M_r$,
      \label{redex-R}
    \item
      $Q \congr \new W Y_s Y_r.
        (S' \parallel R'\subst{x->b} \parallel C)$,
      \label{redex-Q}
  \end{defenum}
  \end{minipage}
  \hfill
  \begin{minipage}{.45\linewidth}
  \begin{defenum}[resume=semantics]
    \item
      $P \congr \new W.(\tact.\new Y.P' \parallel C) \in \PiNf$,
      \label{redex-tau-P}
    \item
      $Q \congr \new W Y. (P' \parallel C)$.
      \label{redex-tau-Q}
  \end{defenum}
  \end{minipage}\\[1ex]
  We define the set of reachable terms from $P$ as
    $\reach(P) \is \Set{Q | P \redto^* Q}$,
  writing $\redto^*$ to mean the reflexive, transitive closure of $\redto$.
  We refer to the restrictions, $\new Y_s$, $\new Y_r$ and $\new Y$,
  as the restrictions \emph{activated} by the transition $P \redto Q$.
\end{definition}
Notice that the use of structural congruence in the definition of $\redto$
takes unfolding replication into account.

\begin{example}[Client-server]
\label{ex:server-client}%
\label{ex:servers}%
  We can model a variation of the client-server pattern
  sketched in the introduction, with the term
    $ \new s\:c.P $
  where
    $P = \bang{S} \parallel \bang{C} \parallel \bang{M}$,
    $S = s(x).\new d.\out x<d> $,
    $C = c(m).(\out s<m> \parallel \inp m(y).\out c<m>) $ and
    $M = \tact.\new m.\out c<m>$.
  The term $\bang{S}$ represents a server
  listening to a port $s$ for a client's requests.
  A request is a channel $x$ that the client sends to the server
  for exchanging the response.
  After receiving $x$ the server creates a new name $d$ and sends it over $x$.
  The term $\bang{M}$ creates unboundedly many clients,
  each with its own private mailbox $m$.
  A client on a mailbox $m$ repeatedly sends requests to the server
  and concurrently waits for the answer on the mailbox before recursing.
\end{example}

In the following examples, we use \CCS{}-style nullary channels, which can be understood as a shorthand:
$\inpz{c}.P \is \inp c(x).P$ and $\outz{c}.P \is \new x.\out c<x>.P$
where $x \not\in \freenames(P)$.

\begin{example}[Resettable counter]%
\label{ex:counter}
A counter with reset is a process reacting to messages on three channels
$\ch{inc}$, $\ch{dec}$ and $\ch{rst}$.
An $\ch{inc}$ message increases the value of the counter,
a $\ch{dec}$ message decreases it or causes a deadlock if the counter is zero, and
a $\ch{rst}$ message resets the counter to zero.
This behaviour is exhibited by the process
$
  C_i = \Bang{
    \inp p_i(t).\bigl(
        \inpz \ch{inc}_i.(\outz t \parallel \out p_i<t>)
      + \inpz \ch{dec}_i.(\inpz t.\out p_i<t>)
      + \inpz \ch{rst}_i.(\new t'_i.\out p_i<t'_i>)
    \bigr)
  }
$.
Here, the number of processes $\outz t$ in parallel with $\out p_i<t>$
represents the current value of the counter $i$.
A system
$ \bigl(
  \new p_1\:t_1.(C_1 \parallel \out p_1<t_1>)
    \parallel
  \new p_2\:t_2.(C_2 \parallel \out p_2<t_2>)
\bigr) $
can for instance simulate a two-counter machine
when put in parallel with a finite control process sending signals along the channels
$\ch{inc}_i$, $\ch{dec}_i$ and $\ch{rst}_i$.
\end{example}

\begin{example}[Unbounded ring]
\label{ex:ring}
  Let
    $R = \new m.\new s_0.(
                M \parallel \out m<s_0> \parallel \outz{s_0}
              )$,
    $S = \bang{(\inpz{s}.\outz{n})}$ and
    $M = \Bang{
                \inp m(n).
                  \inpz s_0.
                    \new s.(
                      S \parallel \out m<s> \parallel \outz{s}
                    )
              } $.
  The term $R$ implements an unboundedly growing ring.
  It initialises the ring with a single \enquote{master} node
  pointing at itself ($s_0$) as the next in the ring.
  The term $M$, implementing the master node's behaviour,
  waits on $\ch{s_0}$ and reacts to a signal
  by creating a new slave with address $s$
  connected with the previous next slave $n$.
  A slave $S$ simply propagates the signals on its channel
  to the next in the ring.
\end{example}

\subsection{Forest representation of terms}
\label{sec:forests}

In the technical developement of our ideas,
we will manipulate the structure of terms in non-trivial ways.
When reasoning about these manipulations,
a term is best viewed as a forest representing
(the relevant part of)
its abstract syntax tree.
Since we only aim to capture the active portion of the term,
the active sequential subterms are the leaves of its forest view.
Parallel composition corresponds to (unordered) branching,
and names introduced by restriction are represented
by internal (non-leaf) nodes.

A \emph{forest} is a simple,
acyclic, directed graph, $f = (N_f, \parent_f)$,
where the edge relation $n_1 \parent_f n_2$ means ``$n_1$ is the parent of $n_2$''.
We write $\parent*_f$ and $\parent+_f$ for the
reflexive transitive and the transitive closure
of $\parent_f$ respectively.
A \emph{path} is a sequence of nodes, $n_1 \, \dots \, n_k$,
such that for each $i < k$, $n_{i} \parent_f n_{i+1}$.
Henceforth we drop the subscript $f$ from
  $\parent_f, \parent*_f$ and $\parent+_f$
(as there is no risk of confusion),
and assume that all forests are finite.
Thus every node has a unique path to a root
(and that root is unique).

An \emph{$L$-labelled forest}
is a pair $\phi = (f_\phi, \ell_\phi)$ where
  $f_\phi$ is a forest and
  $\ell_\phi \from N_\phi \to L$ is a labelling function on nodes.
Given a path $n_1 \dots n_k$ of $f_\phi$, its \emph{trace} is the induced sequence
  $\ell_\phi(n_1) \dots \ell_\phi(n_k)$.
By abuse of language, a \emph{trace} is an element of $L^\ast$
which is the trace of some path in the forest.

We define $L$-labelled forests inductively from the empty forest $\emptyforest$.
We write $\phi_1\dunion \phi_2$ for the disjoint union
of forests $\phi_1$ and $\phi_2$,
and $l[\phi]$ for the forest with a single root,
which is labelled with $l \in L$,
and whose children are the respective roots of the forest $\phi$.
Since the choice of the set of nodes is irrelevant, we will always
interpret equality between forests up to isomorphism
(i.e. a bijection on nodes respecting parent and labeling).

\begin{definition}[Forest representation]
\label{def:forest-repr}
  We represent the structural congruence class of a term $P \in \PiTerms$
  with the set of labelled forests $\AST{P} \is \set{\forest(Q) | Q \congr P}$
  with labels in $\actrestr(P) \dunion \seqproc(P)$
  where $\forest(Q)$ is defined as
  \[
    \forest(Q) \is
      \begin{cases}
        \emptyforest                      \CASE Q = \zero\\
        Q[\emptyforest]                   \CASE Q \text{ is sequential}\\
        x[\forest(Q')]                    \CASE Q = \new x.Q' \\
        \forest(Q_1) \dunion \forest(Q_2) \CASE Q = Q_1 \parallel Q_2
      \end{cases}
  \]
  Note that leaves (and only leaves) are labelled with sequential processes.

  The \emph{restriction height}, $\height_\restr(\forest(P))$, is the length
  of the longest path formed of nodes labelled with names in $\forest(P)$.
\end{definition}

\begin{figure}[tb]
  \centering%
  \input{export/example-forests.patch.tex}
  \caption[Examples of forests in $\AST{P}$]{%
    Examples of forests in $\AST{P}$ 
    where $ P = \new a\:b\:c. ( A_1 \parallel A_2 \parallel A_3 \parallel A_4 ) $,
      $A_1 = \inp a(x)$,
      $A_2 = \inp b(x)$,
      $A_3 = \inp c(x)$ and
      $A_4 = \out a<b>$.%
  }%
  \label{fig:forests}
\end{figure}

\noindent
In \cref{fig:forests} we show some of the possible forest representations
of an example term.

\subsection{Depth-bounded terms}
\label{sec:db}

\begin{definition}[Depth-bounded term~\cite{Meyer:08}]
\label{def:nest}\label{def:depth}
\label{def:depth-bounded}
  The \emph{nesting of restrictions} of a term is given by the function
  \begin{align*}
    \nestr(M) & \is \nestr(\bang{M}) \is \nestr(\zero) \is 0 \\
    \nestr(\new x.P) & \is 1 + \nestr(P) \\
    \nestr(P \parallel Q) & \is \max(\nestr(P), \nestr(Q)).
  \end{align*}
  The \emph{depth} of a term is defined as
  the minimal nesting of restrictions in its congruence class,
  $
    \depth(P) \is \min\set{\nestr(Q) | P \congr Q}.
  $
  A term $P \in \PiTerms$ is \emph{depth-bounded} if
  there exists $k \in \Nat$ such that
  for each $Q \in \reach(P)$, $\depth(Q) \leq k$.
  We write $\DBTerms$ for the set of terms with bounded depth.
\end{definition}

\iflongversion
  Notice that $\nestr$ is \emph{not} an invariant of structural congruence,
  whereas $\depth$ and depth-boundedness are.
  \begin{example}
    Consider the congruent terms $P$ and $Q$
    \[
    P =
      \new a . \new b . \new c . \big(
          \inp a(x) \parallel \out b<c> \parallel \inp c(y)
        \big)
    \congr
      \new a . \inp a(x) \parallel
      \new c . \big(
        (\new b . \out b<c>)
          \parallel
        \inp c(y)
      \big)
    = Q
    \]
    We have $\nestr(P) = 3$ and $\nestr(Q) = 2$;
    but $\depth(P) = \depth(Q) = 2$.
  \end{example}
\fi
It is straightforward to see that the nesting of restrictions of a term coincides with the height of its forest representation,
i.e.,~for every $P \in \PiTerms$, $\nestr(P) = \height_\restr(\forest(P))$.

\begin{example}[Depth-bounded term]
\label{ex:bounded}
  The term in \cref{ex:server-client} is depth-bounded:
  all the reachable terms are congruent to terms of the form
  \[
    Q_{ijk} = \new s\:c.\bigl(
      P \parallel N^i \parallel \mathit{Req}^j \parallel \mathit{Ans}^k
    \bigr)
  \]
  for some $i, j, k \in \Nat$
  where $N = \new m.\out c<m>$,
        $\mathit{Req} =\new m.(\out s<m> \parallel \inp m(y).\out c<m>)$ and
        $\mathit{Ans} = \new m.(\new d.\out m<d> \parallel \inp m(y).\out c<m>)$.
  For any $i, j, k$, $\nestr(Q_{ijk}) \leq 4$.%
\end{example}
\begin{example}[Depth-unbounded term]
\label{ex:unbounded}
  Consider the term in \cref{ex:ring} and the following run%
  :
  \begin{align*}
    R & \redto^\ast
      \new m\:s_0.(
        M \parallel
        \new s_1.(
          \bang{(\inpz{s_1}.\outz{s_0})} \parallel
          \out  m<s_1> \parallel \outz{s_1}
        ))
        \\
      & \redto^\ast
      \new m\:s_0.(
        M \parallel
        \new s_1.(
          \bang{(\inpz{s_1}.\outz{s_0})} \parallel
          \new s_2.(
            \bang{(\inpz{s_2}.\outz{s_1})} \parallel
            \out  m<s_2> \parallel \outz{s_2}
          ))
        )
        \redto^\ast \ldots
  \end{align*}
  The scopes of $s_0$, $s_1$, $s_2$ and the rest
  of the instantiations of $\restr s$
  are inextricably nested,
  thus $R$ has unbounded depth:
  for each $n \geq 1$, a term with depth $n$ is reachable.
\end{example}

%

Depth boundedness is a semantic notion.
Because the definition is a universal quantification over reachable terms, analysis of depth boundedness is difficult.
Indeed the membership problem is undecidable~\cite{Meyer:phd}.
In the communication topology interpretation, depth has a tight
relationship with the maximum length of the simple paths.
A path $v_1 e_1 v_2 \dots v_n e_n v_{n+1}$ in $\CommTop{P}$ is \emph{simple}
if it does not repeat hyper-edges, i.e.,~$e_i \neq e_j$ for all $i \neq j$.
A term is depth-bounded if and only if
there exists a bound on the length of the simple paths of
the communication topology of each reachable term~\cite{Meyer:08}.
This allows terms to grow unboundedly in \emph{breadth},
i.e.,~the degree 
of hyper-edges in the communication topology.

A term $P$ is \emph{embeddable} in a term $Q$, written $P \embeddedin Q$,
if
  $ P \congr \new X.\Parallel_{i \in I} A_i \in \PiNf $
and
  $ Q \congr \new X Y.(\Parallel_{i \in I} A_i \parallel R) \in \PiNf $
for some term $R$.
In~\cite{Meyer:08} the term embedding ordering, $\embeddedin$, is shown to be both
  a simulation relation on \piterm{s},
  and an effective well-quasi ordering on depth-bounded terms.
This makes the transition system
  $(\reach(P)/_{\congr}, {\redto}/_{\congr}, P)$
a \emph{well-structured transition system} (WSTS)~\cite{Finkel:01,Abdulla:96}
under the term embedding ordering.
Consequently a number of verification problems
are decidable for terms in $\DBTerms$.

\begin{theorem}[Decidability of termination~\cite{Meyer:08}]
  \label{th:db-term-decidable}
  The termination problem for depth-bounded terms,
  which asks, given a term $P_0 \in \DBTerms$,
  if there is an infinite sequence
  $P_0 \redto P_1 \redto \ldots$,
  is decidable.
\end{theorem}

\begin{theorem}[Decidability of coverability~\cite{Meyer:08,Wies:10}]
\label{th:db-cover-decidable}
  The coverability problem for depth-bounded terms, which asks,
  given a term $P \in \DBTerms$ and a \emph{query} $Q \in \PiTerms$,
  if there exists $P' \in \reach(P)$ such that $Q \embeddedin P'$,
  is decidable.
\end{theorem}

\section{\tcompatibility\ and hierarchical terms}
\label{sec:t-compat}

A hierarchy is specified by a finite forest $(\Types, \parent)$.
In order to formally relate \emph{active} restrictions in a term to nodes of the hierarchy \Types,
we annotate restrictions with types.
For the moment we view types abstractly as elements of a set $\HTypes$, equipped with a map $\base \from \HTypes \to \Types$.
An annotated restriction $\restr (x\tas \type)$ where $\type \in \HTypes$ will be associated with the node $\base(\type)$ in the hierarchy \Types.
Elements of \HTypes\ are called \emph{types}, and those of \Types\ are called \emph{base types}.
In the simplest case and, especially
for Section~\ref{sec:t-compat},
we may assume $\HTypes = \Types$ and $\base(t) = t$.
In \cref{sec:typesys} we will 
consider a set $\HTypes$ of types generated from $\Types$, and a non-trivial $\base$ map.
\begin{definition}[Annotated term]%
\donotbreak%
\label{def:annot-term}%
A \emph{\pre\HTypes-annotated \piterm}
  (or simply \emph{annotated \piterm}) $P \in \PiAnnot$
  has the same syntax as ordinary \piterm{s}
  except that restrictions take the form
    $\restr (x \tas \type)$
  where $\type \in \HTypes$.
  In the abbreviated form $\restr X$,
  $X$ is a set of 
  annotated names $(x\tas \type)$.
 \end{definition}
\emph{Structural congruence}, $\congr$, of annotated terms, is defined by Definition~\ref{def:congr}, with the proviso that the type annotations are invariant under \pre\alpha-conversion and replication.
For example,
$
  \Bang{\pi . \new (x\tas\tau) . P}
    \congr
  \pi . \new (x\tas\tau) . P \parallel \Bang{\pi . \new (x\tas\tau) . P}
$
and
$
  \new (x \tas \type) . P  \congr  \new (y \tas \type) . P\subst{x -> y}
$;
observe that the annotated restrictions that occur in a replication unfolding
are necessarily inactive.

The forest representation of an annotated \piterm{}
  is obtained from Definition~\ref{def:forest-repr}
  by replacing the case of $Q = \new (x \tas \type).Q'$ by
  \[
    \forest(\new (x \tas \type).Q') := (x, t)[\forest(Q')]
  \]
  where $\base(\type) = t$.
  Thus the forests in $\AST{P}$ have labels in
  $(\actrestr(P) \times \Types) \dunion \seqproc(P)$.
  We write $\Forests_{\Types}$ for the set of forests with labels in
    $(\Names \times \Types) \dunion \Seq$.
  We write $\PiNfAnnot$ for the set of \pre\HTypes-annotated \piterm{s}
  in normal form.

The definition of the transition relation of annotated terms,
  $P \redto Q$,
  is obtained from Definition~\ref{def:sem-of-picalc},
  where $W, Y_s, Y_r$ and $Y$ are now sets of annotated names,
  by replacing clauses \ref{redex-Q} and \ref{redex-tau-Q} by
  \\[1ex]
  \begin{minipage}{.5\linewidth}\itshape
    \begin{defenum}[leftmargin=3em]
      \item[(iv')] $Q \congr \new W Y_s' Y_r'.
              (S' \parallel R'\subst{x->b} \parallel C)$
    \end{defenum}
  \end{minipage}
  \hfill
  \begin{minipage}{.45\linewidth}\itshape
    \begin{defenum}[rightmargin=3em]
      \item[(vi')] $Q \congr \new W Y'. (P' \parallel C)$
    \end{defenum}
  \end{minipage}\\[1ex]
  respectively, such that
    $Y_s \restriction \Names = Y_s' \restriction \Names$,
    $Y_r \restriction \Names = Y_r' \restriction \Names$, and
    $Y \restriction \Names = Y \restriction \Names$,
  where $X \restriction \Names := \set{x \in \Names \mid \exists \type . (x:\type) \in X}$.
  I.e.~the type annotation of the names that are activated by the transition (i.e.~those from $Y_s, Y_r$ and $Y$) are not required to be preserved in $Q$.
  (By contrast, the annotation of every active restriction in $P$ is preserved by the transition.)
  While in this context inactive annotations can be ignored by the transitions,
  they will be used by the type system in \cref{sec:typesys},
  to establish invariance of \tcompat.


Now we are ready to explain what it means
for an annotated term $P$ to be {\tcompat}:
there is a forest in $\AST{P}$ such that
every trace of it projects to a chain in the partial order $\Types$.

\begin{definition}[{\tcompat[ibility]}]
\label{def:t-compat}
  Let $P \in \PiAnnot$ be an annotated \piterm.
  A forest $\phi \in \AST{P}$ is \emph{\tcompat} if
  for every trace
    $((x_1, t_1) \dots (x_k, t_k) \: A)$
  in $\phi$ it holds that
    $t_1 \tlt t_2 \tlt \dots \tlt t_k$.
  The \piterm\ $P$ is \emph{\tcompat} if
  $\AST{P}$ contains a \tcompat\ forest.
  A term is \emph{\tshaped}\ if each of its subterms is \tcompat.
\end{definition}

As a property of annotated terms,
\tcompat[ibility] is by definition invariant under structural congruence.


A term $P'\in\PiAnnot$ is a \emph{type annotation} (or simply \emph{annotation}) of $P\in\PiTerms$
if its \emph{type-erasure}, written $\erase{P}$, coincides with $P$.
(We omit the obvious definition of type-erasure.)
A \emph{consistent annotation} of a transition of terms, $P \redto Q$, is a choice function that, given an annotation $P'$ of $P$, returns an annotation $Q'$ of $Q$ such that $P' \redto Q'$.
Note that it follows from the definition that the annotation of every active restriction in $P'$ is preserved in $Q'$.
The effect of the choice function is therefore to pick a possibly new annotation for each restriction in $Q'$ that is activated by the transition.
Thus, given a semantics $(\PiTerms, \redto, P)$ of a term $P$, and an annotation $P'$ of $P$, and a consistent annotation for every transition of the semantics, there is a well-defined pointed transition system $(\PiAnnot, \redto', P')$ such that every transition sequence of the former lifts to a transition sequence of the latter.
We call  $(\PiAnnot, \redto', P')$ a \emph{consistent annotation} of the semantics $(\PiTerms, \redto, P)$.

\begin{definition}[Hierarchical term]
\label{def:hierarchical}
  A term $P\in\PiTerms$ is \emph{hierarchical}
  if there exist a finite forest $\Types = \HTypes$
  and a consistent annotation $(\PiAnnot, \redto', P')$
  of the semantics $(\PiTerms, \redto, P)$ of $P$,
  such that all terms reachable from $P'$ are \tcompat.
\end{definition}

\begin{example}
The term in \cref{ex:servers,ex:bounded} is hierarchical:
take the hierarchy
  $\Types = \ty{s} \parent \ty{c} \parent \ty{m} \parent \ty{d}$
and annotate each name in $Q_{ijk}$ as follows:
  $ s \tas \ty{s} $,
  $ c \tas \ty{c} $,
  $ m \tas \ty{m} $ and
  $ d \tas \ty{d} $.
The annotation is consistent, and $\forest(Q_{ijk})$ is \tcompat\ for all $i$,$j$ and $k$.

\cref{ex:ring} gives an example of a term that is not hierarchical.
The forest representation of the reachable terms
shown in \cref{ex:unbounded} does not have a bounded height,
which means that if $\Types$ has $n$ base types,
there is a reachable term with a representation of height bigger than $n$,
which implies that there will be a path repeating a base type.
\end{example}

Let us now study this fragment.
First it is easy to see that invariance of \tcompat[ibility] under reduction $\redto$,
  for some finite \Types,
puts a bound $\card{\Types}$ on the height of the \tcompat\ reachable forests,
and consequently a bound on depth.

\begin{theorem}
\label{th:hier-db}
  Every hierarchical term is depth-bounded. The converse is false.
\end{theorem}

Thanks to \cref{th:db-cover-decidable},
an immediate corollary of \cref{th:hier-db} is that
coverability and termination are decidable for hierarchical terms.

\smallskip



Unfortunately, like the depth-bounded fragment, membership of the hierarchical fragment is undecidable.
The proof is by adapting the argument for the undecidability of depth boundedness~\cite{Meyer:phd}.

\begin{lemma}
\label{lemma:term-hier}
  Every terminating \piterm{} is hierarchical.
\end{lemma}

\begin{proof}
  Since the transition system of a term,
    quotiented by structural congruence,
  is finitely branching,
  by K\"onig's lemma the computation tree of a terminating term is finite,
  so it contains finitely many reachable processes
  and therefore finitely many names.
  Take the set of all (disambiguated) active names of the reachable terms and
  fix an arbitrary total order $\Types$ on them.
  The consistent annotation with $(x:x)$ for each name
  will prove the term hierarchical.
\end{proof}

\begin{theorem}
\label{th:hier-undec}
  Determining whether an arbitrary \piterm{} is hierarchical, is undecidable.
\end{theorem}

\begin{proof}
  The \picalc{} is Turing-complete, so termination is undecidable.
  Suppose we had an algorithm to decide if a term is hierarchical.
  Then we could decide termination of an arbitrary \piterm{} by
  first checking if the term is hierarchical;
  if the answer is yes,
    we can decide termination for it by~\cref{th:db-term-decidable},
  otherwise
    we know that it is not terminating by \cref{lemma:term-hier}.
\end{proof}

\Cref{th:hier-undec}---and the corresponding version for depth-bounded terms---is a serious impediment to any practical application of hierarchical terms to verifcation:
when presented with a term to verify,
one has to prove that it belongs to one of the two fragments, manually,
before one can apply the relevant algorithms.

While the two fragments have a lot in common,
hierarchical systems have a {richer} structure,
which we will exploit to define a type system that
can prove a term hierarchical,
in a feasible, sound but incomplete way.
Thanks to the notion of hierarchy, we are thus able to statically capture an expressive fragment of the $\pi$-calculus that enjoys decidable coverability.%

\section{A type system for hierarchical topologies}
\label{sec:typesys}

The purpose of this section is to devise a static check to determine
if a term is hierarchical.
To do so, we define a type system, parametrised over a forest~\Types,
which satisfies subject reduction.
Furthermore we prove that if a term is typable
then \tshaped[ness] is preserved by reduction of the term.
Typability together with \tshaped[ness] of the initial term would then
prove the term hierarchical.

As we have seen in the introduction,
the typing rules make use of a new perspective on \picalc{} reactions.
Take the term
\ifshortversion
  $ \new a. (\new b.\out a<b>.S \parallel \new c.\inp a(x).R) $
  where $ \new a . (\new b. [-] \parallel \new c. [-]) $
\else
  \[
    P = \new a. (\new b.\out a<b>.S \parallel \new c.\inp a(x).R)
      = C[\out a<b>.S, \inp a(x).R]
  \]
  where
  $ C[-_1, -_2] = \new a . (\new b. [-_1] \parallel \new c. [-_2]) $
\fi
is the \emph{reaction context}.
Standardly the synchronisation of the two sequential processes over $a$
is preceded by an extrusion of the scope of $b$
to include $\new c.\inp a(x).R$,
followed by the actual reaction:
\begin{align*}
  \new a. \big( \new b. (\out a<b>.S) \parallel \new c.\inp a(x).R \big)
    &\congr
  \new a.\new b.\big( \out a<b>.S \parallel \new c.\inp a(x).R \big)
  \\ &
    \redto
  \new a.\new b. \big(S \parallel \new c.(R\subst{x -> b}) \big)
\end{align*}%
This dynamic reshuffling of scopes is problematic for establishing
invariance of \tcompat[ibility] under reduction:
notice how $\restr c$ is brought into the scope of $\restr b$,
possibly disrupting \tcompat[ibility].
(For example, the preceding reduction would break \tcompat[ibility] of the forest representations if the tree $\Types$ is either $a \parent c \parent b$ or $b \rparent a \parent c$.)
We therefore adopt a different view.
After the message is transmitted, the sender continues in-place as $S$,
while $R$ is split into two parts
  $R_{\text{mig}} \parallel R_{\neg\text{mig}}$,
one that uses the message
(the \emph{migratable} one) and one that does not.
The migratable portion $R_{\text{mig}}$ is \enquote{installed} under $\restr b$
so that it can make use of the acquired name,
while the non-migratable one can simply continue in-place:%
\ifshortversion
  \[
    \new a. \big( \new b.(\out a<b>.S) \parallel \new c . \inp a(x).R \big) \redto
    \new a. \big( \new b.(S \parallel R_{\text{mig}}\subst{x -> b} )
                \parallel \new c.R_{\neg\text{mig}} \big)
  \]
  Crucially, the \emph{reaction context},
    $\new a . (\new b. [-] \parallel \new c. [-])$,
  is unchanged.
\else
  \[
    \underbrace{
    \new a. \big( \new b.(\out a<b>.S) \parallel \new c . \inp a(x).R \big)
    }_{C[\out a<b>.S,\ \inp a(x).R ]}
      \redto
    \underbrace{
    \new a. \big( \new b.(S \parallel R_{\text{mig}}\subst{x -> b} )
                \parallel \new c.R_{\neg\text{mig}} \big)
    }_{C[ S \parallel R_{\text{mig}}\subst{x -> b},\ R_{\neg\text{mig}} ]}
  \]
  Crucially, the reaction context $C$ is unchanged.
\fi
This means that if the starting term is \tcompat,
the context of the \emph{reactum} is \tcompat\ as well.
Naturally, this only makes sense
if $R_{\text{mig}}$ does not use $c$.
Thus our typing rules impose constraints on the use of names of $R$ so that the migration does not result in $R_{\text{mig}}$ escaping the scope of bound names such as $c$.

The formal definition of ``migratable'' is subtle.
Consider the term
\[
  \new f.\inp a(x).\new c\:d\:e.\bigl(
    \out x<c> \parallel \out c<d> \parallel \out a<e>.\out e<f>
  \bigr)
\]
Upon synchronisation with $\new b.\out a<b>$,
surely $\out x<c>$ will need to be put under the scope of $\restr b$
after substituting $b$ for $x$,
hence the first component of the continuation, $\out x<c>$, is migratable.
However this implies that the scope of $\restr c$
will need to be placed under $\restr b$,
which in turn implies that $\out c<d>$
needs to be considered migratable as well.
On the other hand,
  $\new e.\out a<e>.\out e<f>$
must be placed in the scope of $f$, which may not be known by the sender,
so it is not considered migratable.%
The following definition makes these observations precise.

\begin{definition}[Linked to, tied to, migratable]
  Given a normal form $P = \new X. \Parallel_{i \in I} A_i$ we say that
  $A_i$ is \emph{linked to $A_j$ in $P$}, written $i \linkedto{P} j $, if
  $
    \freenames(A_i) \inters
    \freenames(A_j) \inters
    X
      \neq \emptyset
  $.
  We define the \emph{tied-to} relation as
  the transitive closure of $\linkedto{P}$.
  I.e.~$A_i$ is \emph{tied to} $A_j$, written $i \tiedto{P} j$, if
  $\exists \lstc{k}{n} \in I \st i \linkedto{P} k_1 \linkedto{P} k_2 \ldots \linkedto{P} k_n \linkedto{P} j$, for some $n \geq 0$.
  Furthermore, we say that a name $y$ is \emph{tied to $A_i$ in $P$},
  written $y \ntiedto{P} i$,
  if $\exists j \in I \st y \in \freenames(A_j) \, \wedge \, j \tiedto{P} i$.
  Given an input-prefixed normal form $\inp a(y). P$
  where $P = \new X. \Parallel_{i \in I} A_i$,
  we say that \emph{$A_i$ is migratable in $\inp a(y). P$},
  written $\migr{\inp a(y).P}(i)$, if
  $y \ntiedto{P} i$.
\end{definition}

These definitions have an intuitive meaning
with respect to the communication topology of a normal form $P$:
two sequential subterms are linked
if they are connected by an hyperedge in the communication topology of $P$,
and are tied to each other if there exists a path between them.

The following lemma indicates how the tied-to relation
fundamentally constrains the possible shape of the forest of a term.

\begin{lemma}
\label{lemma:tied-tree}
  Let $P = \new X.\Parallel_{i\in I} A_i \in \PiNf$, if $i \tiedto{P} j$
  then if a forest $\phi \in \AST{P}$ has leaves labelled with $A_i$ and $A_j$ respectively, they belong to the same tree in $\phi$ (i.e., have a common ancestor in $\phi$).
\end{lemma}


\begin{example}
\label{ex:tied-to}
  Take the normal form
    $ P = \new a\:b\:c. ( A_1 \parallel A_2 \parallel A_3 \parallel A_4 ) $
    where
      $A_1 = \inp a(x)$,
      $A_2 = \inp b(x)$,
      $A_3 = \inp c(x)$ and
      $A_4 = \out a<b>$.
  We have $1 \linkedto{P} 4$, $2 \linkedto{P} 4$,
  therefore $1 \tiedto{P} 2 \tiedto{P} 4$ and $a \ntiedto{P} 2$.
  In \cref{fig:forests} we show some of the forests in $\AST{P}$.
  Forest~\ref{forest:broom} represents $\forest(P)$.
  The fact that $A_1, A_2$ and $A_4$ are tied is reflected by the fact
  that none of the forests place them in disjoint trees.
  Now suppose we select only the forests in $\AST{P}$
  that respect the hierarchy $a \parent b$: in all the forests in this set,
  the nodes labelled with $A_1, A_2$ and $A_4$ have $a$ as common ancestor
  (as in forests~\ref{forest:broom},
                 \ref{forest:cab},
                 \ref{forest:min-abc} and
                 \ref{forest:min-ab-c}).
  In particular, in these forests $A_2$ is necessarily a descendent of $a$
  even if $a$ is not one of its free names.
\end{example}

In \cref{sec:t-compat} we introduced annotations in a rather abstract way
by means of a generic domain of types $\HTypes$.
In \cref{def:hierarchical} we ask for the existence of an annotation for the semantics of a term.
Specifically, one can decide an arbitrary annotation for each active name.
A type system however will examine the term statically, which means that it needs to know what could be a possible annotation for a variable, i.e., the name bound in an input action.
This information is directly related to the notion of data-flow,
that is the set of names that are bound to a variable during runtime.
Since a static method cannot capture this information precisely,
we make use of \emph{sorts} \cite{Milner:93},
also known as \emph{simple types},
to approximate it.
The annotation of a restriction will carry
not only which base type should be associated with its instances,
but also instructions on how to annotate the messages
received or sent through those instances.
Concretely, we define
\begin{grammar}
    \HTypes \ni \type \is t \mid t[\type]
\end{grammar}
where $t \in \Types$ is a base type.

A name with type $t$ cannot be used as a channel but can be used as a message;
a name with type $t[\type]$ can be used to transmit a name of type $\type$.
We will write $\base(\type)$ for $t$ when $\type = t[\type']$ or $\type = t$.
By abuse of notation we write, for a set of types $X$,
$\base(X)$ for the set of base types of the types in $X$.

As is standard, we keep track of the types of free names
by means of a typing environment.
An environment $\Env$ is a partial map from names to types,
which we will write as a set of \emph{type assignments}, $x \tas \type$.
Given a set of names $X$ and an environment $\Env$,
we write $\Env(X)$ for the set $\set{\Env(x) | x \in X \inters \domain(\Env)}$.
Given two environments $\Env$ and $\Env'$ with
  $\domain(\Env)\inters\domain(\Env') = \emptyset$,
we write $\Env \Env'$ for their union.
For a type environment $\Env$ we define
\[
\minrestr(\Env) \is \set{(x \tas \type) \in \Env |
    \forall (y \tas \type') \in \Env \st \base(\type') \not\tlt \base(\type)}.
\]

A judgement $\Env\types P$ means that $P \in \PiNfAnnot$ can be typed under assumptions $\Env$, over the hierarchy \Types;
we say that $P$ is \emph{typable} if $\Env\types P$ is provable for some $\Env$ and $\Types$.
An arbitrary term $P \in \PiAnnot$ is said to be \emph{typable} if its normal form is.
The typing rules are presented in \cref{fig:typesys}.

\begin{figure}[t]
  \centering
  \input{definitions/typesys}
  \caption[Typing rules]{
    A type system for hierarchical terms.
    The term $P$ stands for $\new X.\protect\Parallel_{i \in I} A_i$.%
  }
  \label{fig:typesys}
\end{figure}

The type system presents several non-standard features.
First, it is defined on normal forms as opposed to general \piterm{s}.
This choice is motivated by the fact that
different syntactic presentations of the same term
may be misleading when trying to analyse the relation
between the structure of the term and~$\Types$.
The rules need to guarantee that a reduction will not break \tcompat[ibility],
which is a property of the congruence class of the term.
As justified by \cref{lemma:tied-tree},
the scope of names in a congruence class may vary,
but the tied-to relation puts constraints on the structure
that must be obeyed by all members of the class.
Therefore the type system is designed around this basic concept,
rather than the specific scoping of any representative
of the structural congruence class.
Second, no type information is associated with the typed term,
only restricted names hold type annotations.
Third, while the rules are compositional,
the constraints on base types have a global flavour
due to the fact that they involve the structure of $\Types$
which is a global parameter of typing proofs.

Let us illustrate intuitively how the constraints enforced by the rules
guarantee preservation of \tcompat[ibility].
Consider the term
\[
  P = \new e\:a. \Bigl(
    \new b.\bigl( \out a<b>.A_0 \bigr)
      \parallel
    \new d.\bigl( \inp a(x).Q   \bigr)
  \Bigr)
\]
with
$Q = \new c. (
        A_1 \parallel A_2 \parallel A_3
      )$,
$A_0 = \inp b(y)$,
$A_1 = \out x<c>$,
$A_2 = \inp c(z).\out a<e>$ and
$A_3 = \out a<d>$.
Let $\Types$ be the forest with
  $t_e \parent t_a \parent t_b \parent t_c$ and $t_a \parent t_d$,
where $t_x$ is the base type of the (omitted) annotation
of the restriction $\restr x$, for $x \in \set{a,b,c,d,e}$.
The reader can check that $\forest(P)$ is \tcompat.

In the traditional understanding of mobility,
we would interpret the communication of $b$ over $a$
as an application of scope extrusion
to include $\new d.\bigl( \inp a(x).Q \bigr)$ in the scope of $b$
and then syncronisation over $a$ with the application of the substitution
$\subst{x->b}$ to $Q$;
note that the substitution is only valid because the scope of $b$
has been extended to include the receiver.

Our key observation is that we can instead interpret this communication
as a migration of the subcomponents of $Q$ that do get their scopes
changed by the reduction, from the scope of the receiver
to the scope of the sender.
For this operation to be sound,
the subcomponents of $Q$ migrating to the sender's scope
cannot use the names that are in the scope of the receiver but not of the sender.

In our specific example,
after the synchronisation between the prefixes $\out a<b>$ and $\inp a(x)$,
$b$ is substituted to $x$ in $A_1$ resulting in the term $A_1' = \out b<c>$
and $A_0, A_1', A_2$ and $A_3$ become active.
The scope of $A_0$ can remain unchanged
as it cannot know more names than before as a result of the communication.
By contrast, $A_1$ now knows $b$ as a result of the substitution $\subst{x->b}$:
$A_1$ needs to migrate under the scope of $b$.
Since $A_1$ uses $c$ as well, the scope of $c$ needs to be moved under $b$;
however $A_2$ uses $c$ so it needs to migrate under $b$ with the scope of $c$.
$A_3$ instead does not use neither $b$ nor $c$ so it can avoid migration
and its scope remains unaltered.

This information can be formalised using the tied-to relation:
on one hand, $A_1$ and $A_2$ need to be moved together
because $1 \tiedto{Q} 2$ and
they need to be moved because $x \ntiedto{Q} 1, 2$.
On the other hand, $A_3$ is not tied to neither $A_1$ nor $A_2$ in $Q$ and
does not know $x$, thus it is not migratable.
After reduction, our view of the reactum is the term
\[
  \new a. \Bigl(
    \new b.\bigl(
      A_0 \parallel
      \new c. ( A_1' \parallel A_2 )
    \bigr)
  \parallel \new d.A_3
  \Bigr)
\]
the forest of which is \tcompat.
Rule~\ref{rule:conf}, applied to $A_1$ and $A_2$,
ensures that $c$ has a base type that can be nested under the one of $b$.
Rule~\ref{rule:in} does not impose constraints on the base types of $A_3$
because $A_3$ is not migratable.
It does however check that the base type of $e$ is an ancestor of the one of $a$,
thus ensuring that both receiver and sender are already in the scope of $e$.
The base type of $a$ does not need to be further constrained since
the fact that the synchronisation happened on it implies that
both the receiver and the sender were already under its scope;
this implies, by \tcompat[ibility] of $P$, that $c$ can be nested under $a$.

We now describe the purpose of the rules of the type system in more detail.
Most of the rules just drive the derivation through the structure of the term.
The crucial constraints are checked by~\ref{rule:conf}, \ref{rule:in} and \ref{rule:out}.

\paragraph{The \ref{rule:out}~rule.}
The main purpose of rule \ref{rule:out} is enforcing types
to be consistent with the dataflow of the process:
the type of the argument of a channel $a$ must agree with the types
of all the names that may be sent over $a$.
This is a very coarse sound over-approximation of the dataflow;
if necessary it could be refined using well-known techniques from the literature
but a simple approach is sufficient here to type interesting processes.

\paragraph{The \ref{rule:conf}~rule.}
Rule~\ref{rule:conf} is best understood imagining the normal form to be typed,
$P$, as the continuation of a prefix $\pi.P$.
In this context a reduction exposes each of the active sequential subterms of $P$ which need to have a place in a \tcompat\ forest for the reactum.
The constraint in \ref{rule:conf} can be read as follows.
A \enquote{new} leaf $A_i$ may refer to names
already present in the forests of the reaction context;
these names are the ones mentioned in both $\freenames(A_i)$ and $\Env$.
Then we must be able to insert $A_i$ so that we can find these names in its path. However, $A_i$ must belong to a tree containing all the names in $X$ that are tied to it in $P$. So by requiring every name tied to $A_i$ to have a base type greater than any name in the context that $A_i$ may refer to, we make sure that we can insert the continuation in the forest of the context without violating \tcompat[ibility].
Note that $\Env(\freenames(A_i))$ contains only types that annotate names both in $\Env$ and $\freenames(A_i)$, that is, names which are not restricted by $X$ and are referenced by $A_i$ (and therefore come from the context).

\paragraph{The \ref{rule:in}~rule.}
Rule~\ref{rule:in} serves two purposes:
on the one hand it requires the type of the messages
that can be sent through $a$ to be consistent with
the use of the variable $x$ which will be bound to the messages;
on the other hand, it constrains the base types of $a$ and $x$ so that
synchronisation can be performed without breaking \tcompat[ibility].

The second purpose is achieved by distinguishing two cases,
represented by the two disjuncts of the condition on base types of the rule.
In the first case, the base type of the message
is an ancestor of the base type of $a$ in $\Types$.
This implies that in any \tcompat\ forest representing $a(x).P$,
the name $b$ sent as message over $a$ is already in the scope of $P$.
Under this circumstance, there is no real mobility,
$P$ does not know new names by the effect of the substitution $\subst{x->b}$,
and the \tcompat[ibility] constraints to be satisfied are in essence unaltered.

The second case is more complicated as it involves genuine mobility.
This case also requires a slightly non-standard feature:
not only do the premises predicate
on the direct subcomponents of an input prefixed term,
but also on the direct subcomponents of the continuation.
This is needed to be able to separate the continuation in two parts:
the one requiring migration and the one that does not.
The situation during execution is depicted in \cref{fig:explain-in-mig}.
The non migratable sequential terms behave exactly
as the case of the first disjunct: their scope is unaltered.
The migratable ones instead are intended to be inserted
as descendents of the node representing the message $b$
in the forest of the reaction context.

\begin{figure}[htb]
  \centering%
  \input{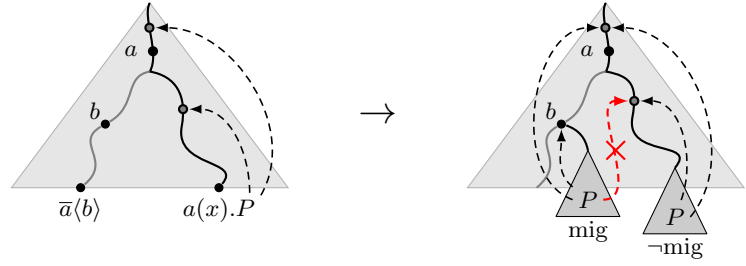}
  \caption[Explanation of constraints imposed by rule~\ref{rule:in}]{%
    Explanation of constraints imposed by rule~\ref{rule:in}.
    The dashed lines represent references to names restricted
    in the reduction context.
  }
  \label{fig:explain-in-mig}
\end{figure}

For this to be valid without rearrangement of the forest of the context,
we need
  all the names in the context that are referenced in the migratable terms,
to be also in the scope at $b$;
we make sure this is the case by requiring
the free names of any migratable $A_i$ that are from the context
(i.e.~in $\Env$) to have base types smaller than the base type of~$a$.
The set
  $\base(\Env(\freenames(A_i)\setminus\set{a}))$
indeed represents the base types of the names
in the reaction context referenced in a migratable continuation $A_i$.
In fact $a$ is a name that needs to be in the scope
of both the sender and the receiver at the same time,
so it needs to be a common ancestor of sender and receiver
in any \tcompat\ forest.
Any name in the reaction context and in the continuation of the receiver,
with a base type smaller than the one of $a$,
will be an ancestor of $a$%
---and hence of the sender, the receiver
   and the node representing the message---%
in any \tcompat\ forest.
Clearly, remembering $a$ is not harmful
as it must be already in the scope of receiver and sender,
so we exclude it from the constraint.

\begin{example}
\label{ex:servers-typing}
  Take the normal form in \cref{ex:server-client}.
  Let us fix $\Types$ to be the forest
    $\ty{s} \parent \ty{c} \parent \ty{m} \parent \ty{d}$
  and annotate the normal form with the following types:
    $ s \tas \type_s = \ty{s}[\type_m] $,
    $ c \tas \type_c = \ty{c}[\type_m] $,
    $ m \tas \type_m = \ty{m}[\ty{d}] $ and
    $ d \tas \ty{d} $.
  We want to prove $\emptyset \types \new s\:c.P$.
  We can apply rule~\ref{rule:conf}:
  in this case there are no conditions on types because,
  being the environment empty,
  we have $\base(\emptyset(\freenames(A))) = \emptyset$
  for every active sequential term $A$ of $P$.
  Let $\Env = \set{(s \tas \type_s), (c \tas \type_c)}$.
  The rule requires
    $\Env\types \bang{S}$, $\Env\types \bang{C}$ and $\Env\types \bang{M}$,
  which can be proved by proving typability of $S$, $C$ and $M$ under $\Env$
  by rule~\ref{rule:bang}.

  To prove $\Env\types S$ we apply rule~\ref{rule:in};
  we have $s \tas \ty{s}[\type_m] \in \Env$
  and we need to prove that $\Env, x\tas\type_m \types \new d.\out x<d>$.
  No constraints on base types are generated at this step
  since the migratable sequential term $\new d.\out x<d>$
  does not contain free variables typed by $\Env$ making
  $\Env(\freenames(\new d.\out x<d>) \setminus \set{a}) = \Env(\set{x})$ empty.
  Next, $\Env, x\tas\type_m \types \new d.\out x<d>$ can be proved by
  applying rule~\ref{rule:conf} which amounts to
  checking $\Env, x\tas\type_m, d\tas\ty{d} \types \out x<d>.\zero$
  (by a simple application of \ref{rule:out}
  and the axiom $\Env, x\tas\type_m, d\tas\ty{d} \types \zero$)
  and verifying the condition---true in $\Types$---%
  $\base(\type_m) \tlt \base(\type_d)$:
  in fact $d$ is tied to $\out x<d>$ and,
  for $\Env' = \Env \union \set{x \tas \type_m}$,
  $
  \base(\Env'(\freenames(\out x<d>)))
    = \base(\Env'(\set{x,d}))
    = \base(\set{\type_m}).
  $
  The proof for $\Env \types M$ is similar and requires
  $\ty{c} \tlt \ty{m}$ which is true in $\Types$.

  Finally, we can prove $\Env \types C$ using rule~\ref{rule:in};
  both the two continuations $A_1 = \out s<m>$ and $A_2 = \inp m(y).\out c<m>$
  are migratable in $C$ and since $\base(\type_m) \tlt \base(\type_c)$ is false
  we need the other disjunct of the condition to be true.
  This amounts to checking that
  $
    \base(\Env(\freenames(A_1) \setminus \set{c}))
      = \base(\Env(\set{s,m})) = \base(\set{\type_s}) = \ty{s} \tlt \ty{c}
  $
  (note $m \not\in \domain(\Env)$)
  and
    $\base(\Env(\freenames(A_a) \setminus \set{c}))
      = \base(\Env(\set{})) \tlt \ty{c}$
  (that holds trivially).

  To complete the typing we need to show
  $\Env, m \tas \type_m \types A_1$ and
  $\Env, m \tas \type_m \types A_2$.
  The former can be proved by a simple application of \ref{rule:out}
  which does not impose further constraints on $\Types$.
  The latter is proved by applying \ref{rule:in} which requires
  $\base(\type_c) \tlt \ty{m}$, which holds in $\Types$.

  Note how, at every step, there is only one rule that applies to each subproof.
\end{example}

\begin{example}
\label{ex:ring-typing}
The term \Cref{ex:ring} is not typable under any \Types.
To see why, one can build the proof tree without assumptions on $\Types$
by assuming that each restriction $\restr x$ has base type $\btyof{x}$.
When typing $\out m<s>$ we deduce that
    $\btyof{s} = \btyof{n}$,
  which is in contradiction with the constraint that
    $\btyof{n} \tlt \btyof{s}$
  required by rule \ref{rule:conf} when typing
  $\new s.(S \parallel \out m<s> \parallel \outz{s})$.
\end{example}

\section{Soundness of the type system}
\label{sec:soundness}

We now establish the soundness of the type system.
\Cref{th:subj-red} will show how typability is preserved by reduction.
\Cref{th:typed-tshaped} establishes the main property of the type system:
if a term is typable then \tshaped[ness] is invariant under reduction.
This allows us to conclude that if a term is \tshaped{} and typable,
then every term reachable from it will be \tshaped{}.

The subtitution lemma states that substituting names
without altering the types preserves typability.
\begin{lemma}[Substitution]
\label{lemma:subst}
  Let $P \in \PiNfAnnot$ and
  $\Env$ be a typing environment such that $\Env(a) = \Env(b)$.
  Then it holds that if\/ $\Env \types P$ then $\Env \types P\subst{a->b}$.
\end{lemma}

Before we state the main theorem, we define the notion of
$P$-safe type environment, which is a simple restriction on the types
that can be assigned to names that are free at the top-level of a term.

\begin{definition}[$P$-safe environment]
A type environment $\Env$ is said to be \emph{$P$-safe} if
for each $x \in \freenames(P)$ and $(y\tas\type) \in \resboundnames(P)$,
$\base(\Env(x)) < \base(\type)$.
\end{definition}

\begin{theorem}[Subject Reduction]
\label{th:subj-red}
  Let $P$ and $Q$ be two terms in $\PiNfAnnot$
  and $\Env$ be a $P$-safe type environment.
  If\/ $\Env \types P$ and $P \redto Q$, then $\Env \types Q$.
\end{theorem}

The proof is by careful analysis of how the typing proof for $P$
can be adapted to derive a proof for $Q$.
The only difficulty comes from the fact that
some of the subterms of $P$ will appear in $Q$ with a substitution applied.
However, typability of $P$ ensures
that we are only substituting names for names with the same type,
thus allowing us to apply \cref{lemma:subst}.

To establish that \tshaped[ness] is invariant under reduction for typable terms,
we will need to show that starting from a typable \tshaped\ term $P$,
any step will reduce it to a (typable) \tshaped\ term.
The hypothesis of \tcompat[ibility] of $P$ can be used
to extract a \tcompat\ forest $\phi$ from $\AST{P}$.
While many forests in $\AST{P}$ can be witnesses
of the \tcompat[ibility] of $P$,
we want to characterise the shape of a witness that \emph{must} exist
if $P$ is \tcompat.
The proof of invariance relies on selecting a $\phi$ that does not impose unnecessary hierarchical dependencies among names.
Such forest is identified by $\Phi_\Types(\nf(P))$:
it is the shallowest among all the \tcompat\ forests in $\AST{P}$.

\begin{definition}[$\Phi_\Types$]
\label{def:phi}
The function $\Phi_\Types \from \PiNfAnnot \to \Forests_\Types$
is defined inductively as
\begin{align*}
  \Phi_\Types(\Parallel_{i \in I} A_i) & \is
      \Dunion_{i\in I} \set{A_i[]} \\
  \Phi_\Types(P) & \is
        \left(
          \Dunion
            \Set{(x,\base(\type))[ \Phi_\Types(\new Y_x.\Parallel_{j \in I_x} A_j) ]
                | (x\tas\type) \in \minrestr(X) }
        \right) \\
      & \quad\dunion
      \Phi_\Types(\new Z.\Parallel_{r\in R} A_r)
\end{align*}
where
  $X \neq \emptyset$,
  $P = \new X.\Parallel_{i \in I} A_i$
  $I_x = \set{i\in I | x \ntiedto{P} i}$
  and
\begin{align*}
  Y_x &= \set{(y \tas \type) \in X | \exists i \in I_x \st y \in \freenames(A_i)} \setminus \minrestr(X) \\
  Z   &= X \setminus \bigr(\textstyle \Union_{(x\tas\type) \in \minrestr(X)} Y_x \union \set{x \tas \type}\bigl) \\
  R   &= I \setminus \bigr(\textstyle \Union_{(x\tas\type) \in \minrestr(X)} I_x\bigl)
\end{align*}
\end{definition}

Forest~\ref{forest:min-ab-c} of \cref{fig:forests} is $\Phi_\Types(P)$
when every restriction $\restr x$ has base type $x$
(for $x \in \set{a,b,c}$)
and $\Types$ is the forest with nodes $a$, $b$ and $c$
and a single edge $a \parent b$.

\begin{lemma}
\label{lemma:phi-tcompat}
  Let $P \in \PiNfAnnot$.
  Then:
  \begin{enumerate}[label=\alph*)]
    \item $\Phi_\Types(P)$ is a \tcompat\ forest;
          \label{lemma:phi-tcompat:tcompat}
    \item $\Phi_\Types(P) \in \AST{P}$ if and only if $P$ is \tcompat;
          \label{lemma:phi-tcompat:ast}
    \item if $P \congr Q \in \PiAnnot$ then
          $\Phi_\Types(P) \in \AST{Q}$ if and only if\/ $Q$ is \tcompat.
          \label{lemma:phi-tcompat:congr}
  \end{enumerate}
\end{lemma}

\begin{theorem}[Invariance of {\tshaped[ness]}]
\label{th:typed-tshaped}
  Let $P$ and $Q$ be terms in $\PiNfAnnot$ such that $P \redto Q$ and
  $\Env$ be a $P$-safe environment such that $\Env \types P$.
  Then, if $P$ is \tshaped{} then $Q$ is \tshaped{}.
\end{theorem}

The key of the proof is
\begin{inparaenum}
  \item the use of $\Phi_\Types(P)$ to extract a specific \tcompat\ forest,
  \item the definition of a way to insert the subtrees of
  the continuations of the reacting processes
  in the forest of reaction context,
  in a way that preserves \tcompat[ibility].
\end{inparaenum}
Thanks to the constraints of the typing rules, we will always be able to find a valid place in the reaction context where to attach the trees representing the reactum.

\section{Type inference}
\label{sec:inference}

\newcommand{\C}[1]{\constraints_{#1}}

In this section we will show that it is possible to take
any non-annotated normal form $P$
and derive a forest $\Types$ and an annotated version of $P$
that can be typed under $\Types$.

Inference for simple types has already
been proved decidable in~\cite{Gay:93,Vasconcelos:93}.
In our case, since our types are not recursive,
the algorithm concerned purely with the constraints imposed by the type system
of the form $\type_x = t[\type_y]$ is even simpler.
The main difficulty is inferring the structure of $\Types$.

Let us first be more specific on assigning simple types.
The number of ways a term $P$ can be annotated with types are infinite,
simply from the fact that types
allow an arbitrary nesting as in $t$, $t[t]$, $t[t[t]]$ and so on.
We observe that, however,
there is no use annotating a restriction with a type with nesting deeper than
the size of the program:
the type system cannot inspect more deeply nested types.
Thanks to this observation we can restrict ourselves to annotations with
bounded nesting in the type's structure.
This also gives a bound on the number of base types that need to appear
in the annotated term.
Therefore, there are only finitely many possible annotations and possible forests under which $P$ can be proved typably hierarchical.
A na\"ive inference algorithm can then enumerate all of them
and type check each.

\begin{theorem}[Decidability of inference]
\label{th:inference-decidable}
  Given a normal form $P \in \PiNf$,
  it is decidable if there exists a finite forest $\Types$,
  a \preTypes-annotated version $P'\in \PiAnnot$ of $P$
  and a \pre P'-safe environment $\Env$ such that
  $P'$ is \tshaped\ and $\Env \types P'$.
\end{theorem}

While enumerating all the relevant forests, annotations and environments
is impractical, more clever strategies for inference exist.

We start by annotating the term with type variables:
each name $x$ gets typed with a type variable $\typevar_x$.
Then we start the type derivation,
collecting all the constraints on types along the way.
If we can find a $\Types$ and type expressions
to associate to each type variable, so that these constraints are satisfied,
the process can be typed under $\Types$.

By inspecting rules~\ref{rule:conf} and \ref{rule:in} we observe that
all the ``tied-to'' and ``migratable'' predicates do not depend on $\Types$
so for any given $P$, the type constraints can be expressed
simply by conjuctions and disjuctions of two kinds of basic predicates:
\begin{enumerate}
\donotbreak%
  \item \emph{data-flow constraints} of the form
          $\typevar_x = t_x[\typevar_y]$
        where $t_x$ is a base type variable;
  \item \emph{base type constraints} of the form
          $\base(\typevar_x) < \base(\typevar_y)$
        which correspond to constraints over the corresponding
        base type variables, e.g.~$t_x < t_y$.
\end{enumerate}
Note that the $P$-safety condition on $\Env$
translates to constraints of the second kind.
The first kind of constraint can be solved using unification in linear time.
If no solution exists, the process cannot be typed.
This is the case of processes that cannot be \emph{simply typed}.
If unification is successful we get a set of equations over base type variables.
Any assignment of those variables to nodes in a suitable forest
that satisfies the constraints of the second kind
would be a witness of typability.
An example of the type inference in action can be found in \appendixorfull.

First we note that if there exists a $\Types$
which makes $P$ typable and \tcompat,
then there exists a $\Types'$ which does the same
but is a linear chain of base types
(i.e. a single tree with no branching).
To see how, simply take $\Types'$ to be any topological sort of $\Types$.

Now, suppose we are presented with a set $\constraints$ of constraints
of the form $t \tlt t'$ (no disjuctions).
One approach for solving them could be based on
reductions to SAT or CLP(FD).
We instead outline a direct algorithm.
If the constraints are acyclic,
i.e.~it is not possible to derive $t \tlt t$ by transitivity,
then there exists a finite forest satisfying the constraints,
having as nodes the base type variables.
To construct such forest, we can first represent the constraints
as a graph with the base type variables as vertices
and an edge between $t$ and $t'$ just when $t \tlt t' \in \constraints$.
Then we can check the graph for acyclicity.
If the test fails, the constraints are unsatisfiable.
Otherwise, any topological sort of the graph will represent a forest
satisfying $\constraints$.

We can modify this simple procedure to support constraints including disjuctions
by using backtracking on the disjuncts.
Every time we arrive at an acyclic assigment,
we can check for \tshaped[ness] (which takes linear time)
and in case the check fails we can backtrack again.

To speed up the backtracking algorithm,
one can merge the acyclicity test with the \tcompat[ibility] check.
Acyclicity can be checked by constructing a topological sort
of the constraints graph.
Every time we produce the next node in the sorting,
we take a step in the construction of $\Phi(P)$ using the fact that
the currently produced node is the minimal base type among the remaining ones.
We can then backtrack as soon as a choice contradicts \tcompat[ibility].

The complexity of the type checking problem is easily seen
to be linear in the size of the program.
This proves,
  in conjuction with the finiteness of the candidate guesses
  for $\Types$ and annotations,
that the type inference problem is in NP.
We conjecture that inference is also NP-hard.

We implemented the above algorithm in a tool
called `James Bound' (\texttt{jb}),
available at \url{http://github.com/bordaigorl/jamesbound}.%

\section{Expressivity and verification}
\label{sec:verification}

\subsection{Expressivity}
\label{sec:expressivity}

Typably hierarchical terms form a rather expressive fragment.
Apart from including common patterns as the client-server one,
they generalise powerful models of computation with decidable properties.

Relations with variants of \CCS{} are the easiest to establish:
\CCS{} can be seen as a syntactic subset of \picalc{}
when including 0-arity channels, which are very easily dealt with
by straightforward specialisations of the typing rules for actions.
One very expressive, yet not Turing-powerful, variant is \CCS!~\cite{He:11}
which can be seen as our \picalc{} without mobility.
Indeed, every \CCS! process is typably hierarchical%
~\cite[Section 11.4]{DOsualdo:15:phd}.

Reset nets can be simulated by using resettable counters as defined in \cref{ex:counter}. The full encoding can be found in~\appendixorfull.
The encoding preserves coverability but not reachability.

\CCS! was recently proven to have decidable reachability~\cite{He:11}
so it is reasonable to ask whether reachability is decidable
for typably hierarchical terms.

We show this is not the case by introducing a weak encoding of Minsky machines
(in \appendixorfull).
The encoding is weak in the sense that not all of the runs
represent real runs of the encoded Minsky machine;
however with reachability one can distinguish
between the reachable terms that are encodings of reachable configurations
and those which are not.
We therefore reduce reachability of Minsky machines
to reachability of typably hierarchical terms.

\begin{theorem}
\label{th:reach-undec}
  The reachability problem is undecidable for (typably) hierarchical terms.
\end{theorem}

\Cref{th:reach-undec} can be used to clearly separate the (typably) hierarchical
fragment from other models of concurrent computation as Petri Nets,
which have decidable reachability and are thus less expressive.

\subsection{Applications}
\label{sec:applications}

Although reachability is not decidable,
coverability is often quite enough to prove non-trivial safety properties.
To illustrate this point, let us consider \cref{ex:servers} again.
In our example, each client waits for a reply reaching its mailbox
before issuing another request;
moreover the server replies to each request with a single message.
Together, these observations suggest
that the mailboxes of each client will contain at most one message at all times.
To automatically verify this property we could use a coverability algorithm
for depth-bounded systems:
since the example is typable, it is depth-bounded and
such algorithm is guaranteed to terminate with a correct answer.
To formulate the property as a coverability problem,
we can ask for coverability of the following query:
$
  \new s\:m.(
    \bang{S} \parallel
    \inp m(y).\out c<m> \parallel
    \new d.\out m<d> \parallel
    \new d'.\out m<d'>
  )
$.
This is equivalent to asking whether a term is reachable that embeds a server connected with a client with a mailbox containing two messages.
The query is not coverable and therefore we proved our property.%
\footnote{To fully prove a bound on the mailbox capacity one may need to also ask another coverability question for the case where the two messages bear the same data-value~$d$.}%

Other examples of coverability properties are variants of secrecy properties.
For instance, the coverability query
$
  \new s\:m\:m'.(
    \bang{S} \parallel
    \inp m(y).\out c<m> \parallel
    \inp m'(y).\out c<m'> \parallel
    \new d.(\out m<d> \parallel \out m'<d>)
  )
$
encodes the property ``can two different clients receive the same message?'',
which cannot happen in our example.

It is worth noting that this level of accuracy for proving such properties automatically is uncommon.
Many approaches based on counter abstraction~\mbox{\cite{Pnueli:02,Emerson:99}}
or CFA-style abstractions~\cite{DOsualdo:13}
would collapse the identities of clients
by not distinguishing between different mailbox addresses.
Instead a single counter is typically used to record
the number of processes in the same control state and of messages.
In our case, abstracting the mailbox addresses away has the effect
of making the bounds on the clients' mailboxes
unprovable in the abstract model.

A natural question at this point is:
how can we go about verifying terms which cannot be typed,
as the ring example?
Coverability algorithms can be applied to untypable terms
and they yield sound results when they terminate.
But termination is not guaranteed,
as the term in question may be depth-unbounded.

However, even a failed typing attempt
may reveal interesting information about the structure of a term.
For instance, in \cref{ex:ring-typing} one may easily see
that the cyclic dependencies in the constraints are caused
by the names representing the ``next'' process identities.
In the general case heuristics can be employed
to automatically identify a minimal set of problematic restrictions.
Once such restrictions are found,
a counter abstraction could be applied \emph{to those restrictions only}
yielding a term that simulates the original one
but introducing some spurious behaviour.
Type inference can be run again on the the abstracted term;
on failure, the process can be repeated,
until a hierarchical abstraction is obtained.
This abstract model can then be model checked instead of the original term,
yielding sound but possibly imprecise results.

\section{Related work}
\label{sec:relwork}

\emph{Depth boundedness} in the \picalc{} was first proposed in~\cite{Meyer:08}
where it is proved that depth-bounded systems
are {well-structured transition systems}.
In~\cite{Wies:10} it is further proved that (forward) coverability is decidable even when the depth bound $k$ is not known \emph{a priori}.
In~\cite{Zufferey:12} an approximate algorithm for computing the \emph{cover set}---an over-approximation of the set of reachable terms---of a system of depth bounded by $k$ is presented.
All these analyses rely on the assumption of depth boundedness and may even require a known bound on the depth to terminate.


Several other interesting fragments of the $\pi$-calculus have been proposed in the literature, such as
  name bounded~\cite{Huechting:13},
  mixed bounded~\cite{Meyer:09}, and
  structurally stationary~\cite{Meyer:09a}.
Typically defined by a non-trivial condition on the set of reachable terms -- a \emph{semantic} property, membership becomes undecidable.
Links with Petri nets via encodings of proper subsets of depth-bounded systems have been explored in \cite{Meyer:09}.
Our type system can prove depth boundedness for processes that are breadth and name unbounded, and which cannot be simulated by Petri nets.
In~\cite{Amadio:02}, Amadio and Meyssonnier consider fragments of the
asynchronous \picalc{} and show that coverability is decidable for
the fragment with no mobility and bounded number of active sequential processes,
via an encoding to Petri nets.
Typably hierarchical systems can be seen as an extension of the result
for a synchronous \picalc{} with unbounded sequential processes
and a restricted form of mobility.

Recently H\"{u}chting et al.~\cite{Huechting:14} proved
several relative classification results between fragments of \picalc.
Using 
Karp-Miller trees,
they presented an algorithm to decide if an arbitrary \piterm\ is
bounded in depth by a given $k$.
The construction is based on an (accelerated) exploration
of the state space of the \piterm,
with non primitive recursive complexity,
which makes it impractical.
By contrast, our type system uses a very different technique
leading to a quicker algorithm, at the expense of precision.
Our forest-structured types can also act as specifications,
offering more intensional information to the user than just a bound $k$.

Our types are based on Milner's sorts for the \picalc~\cite{Milner:93,Gay:93},
later refined into I/O types~\cite{Pierce:93}
and their variants~\cite{Pierce:00}.
Based on these types is a system for termination of \piterm{s}~\cite{Deng:06}
that uses a notion of levels,
enabling the definition of a lexicographical ordering.
Our type system can also be used to determine termination of \piterm{s}
in an approximate but conservative way,
by using it in conjuction with \cref{th:db-term-decidable}.
Because the respective orderings between types of the two approaches
are different in conception, we expect
the terminating fragments isolated by the respective systems to be incomparable.

\section{Future directions}
\label{sec:future}

The type system we presented in Section~\ref{sec:typesys} is conservative:
the use of simple types, for example, renders the analysis context-insensitive.
Although we have kept the system simple so as to focus on the novel aspects, a number of improvements are possible.
First, the extension to the polyadic case is straightforward.
Second, the type system can be made more precise by using subtyping and polymorphism to refine the analysis of control and data flow.
Third, the typing rule for replication introduces a very heavy approximation:
when typing a subterm, we have no information about which other parts of the term (crucially, which restrictions) may be replicated.
By incorporating some information about which names can be instantiated
unboundedly in the types, the precision of the analysis can be greatly improved.
The formalisation and validation of these extensions is a topic of ongoing research.

Another direction worth exploring is the application of this machinery to
heap manipulating programs and security protocols verification.

\paragraph{Acknowledgement.}
We would like to thank Damien Zufferey for helpful discussions
on the nature of depth boundedness,
and Roland Meyer for insightful feedback on a previous version of this paper.

\bibliographystyle{abbrvnat}
\bibliography{abbrev,biblio}

\ifincludeappendix

  \clearpage

  \appendix

  \section{Supplementary Material for Section~\ref{sec:prelim}}

\subsection{Definition and properties of $\nf$}
\label{app:nf}

The function $\nf \from \PiTerms \to \PiNf$,
defined in \cref{fig:nf},
extracts, from a term, a normal form structurally equivalent to it.

\begin{definition}[$\nf(P)$]
  We define the function $\nf \from \PiTerms \to \PiNf$
  as follows:
  \begin{mathpar}
  \nf(\zero)  \is \zero \and
  \nf(\pi.P)  \is \pi.\nf(P) \and
  \nf(\new x.P) \is \new x. \nf(P) \and
  \nf(M + M') \is \nf(M) + \nf(M')
      \and
  \nf(\bang{M}) \is \bang{(\nf(M))}
    \and
  \nf(P \parallel Q) \is
      \begin{cases}
        \nf(P)  \CASE \nf(Q) = \zero \neq \nf(P) \\
        \nf(Q)  \CASE \nf(P) = \zero \\
        \new X_P X_Q. (N_P \parallel N_Q) \CASE
          \nf(Q) = \new X_Q.N_Q,
          \nf(P) = \new X_P.N_P
          \AND
          \actrestr(N_P) = \actrestr(N_Q) = \emptyset
      \end{cases}
\end{mathpar}

  \label{fig:nf}
\end{definition}

\begin{lemma}
  For each $P \in \PiTerms$,
 $P \congr \nf(P)$
\end{lemma}

\begin{proof}
A straightforward induction on $P$.
\end{proof}

\begin{lemma}
\label{lemma:forest-nf}
Let $\phi$ be a forest with labels in $\Names \dunion \Seq$.
Then $\phi = \forest(Q)$ with
  $Q \congr Q_\phi$ where
  \begin{align*}
    Q_\phi &\is \new X_\phi.\Parallel_{(n, A) \in I} A\\
    X_\phi &\is \set{\ell_\phi(n) \in \Names | n \in N_\phi}\\
    I &\is \set{(n, A) \mid \ell_\phi(n) = A \in \Seq}
  \end{align*}
provided
\begin{enumerate}[label=\small\roman*),ref=\thetheorem.\roman*,]
  \item $\forall n \in N_\phi$,
        if $\ell_\phi(n) \in \Seq$ then
        $n$ has no children in $\phi$, and
        \label{lemma:forest-nf:seq-leaf}
  \item $\forall n, n' \in N_\phi$,
        if $\ell_\phi(n) = \ell_\phi(n') \in \Names$
        then $n = n'$, and
        \label{lemma:forest-nf:name-uniq}
  \item $\forall n \in N_\phi$,
        if $\ell_\phi(n) = A \in \Seq$
        then for each $x \in X_\phi \inters \freenames(A)$
          there exists $n' <_\phi n$ such that $\ell_\phi(n') = x$.
          \label{lemma:forest-nf:scoping}
\end{enumerate}
\end{lemma}

\begin{proof}
  \newcommand{\ALLCOND}{%
  \ref{lemma:forest-nf:seq-leaf},
  \ref{lemma:forest-nf:name-uniq} and
  \ref{lemma:forest-nf:scoping}
}
We proceed by induction on the structure of $\phi$.
The base case is when $\phi = \emptyforest$,
for which we have $Q_\phi = \zero$ and $\phi = \forest(\zero)$.

When $\phi = \phi_0 \dunion \phi_1$ we have that
if conditions \ALLCOND hold for $\phi$,
they must hold for $\phi_0$ and $\phi_1$ as well,
hence we can apply the induction hypothesis to them obtaining
$\phi_i\forest(Q_i)$ with $Q_i \congr Q_{\phi_i}$ ($i \in \set{0,1}$).
We have $\phi = \forest(Q_0 \parallel Q_1)$ by definition of $\forest$,
and we want to prove that $Q_0 \parallel Q_1 \congr Q_{\phi}$.
By condition \ref{lemma:forest-nf:name-uniq} on $\phi$,
$X_{\phi_0}$ and $X_{\phi_1}$ must be disjoint;
furthermore, by condition \ref{lemma:forest-nf:scoping}
on both $\phi_0$ and $\phi_1$ we can conclude that
$\freenames(Q_{\phi_i}) \inters X_{\phi_{1-i}} = \emptyset$.
We can therefore apply scope extrusion:
$Q_0 \parallel Q_1
 \congr Q_{\phi_0} \parallel Q_{\phi_1}
 \congr \new X_{\phi_0} X_{\phi_1}.(P_{\phi_0} \parallel P_{\phi_1})
 = Q_\phi$.

The last case is when $\phi = l[\phi']$.
Suppose conditions \ALLCOND hold for $\phi$.
We distinguish two cases.
If $l = A \in \Seq$, by \ref{lemma:forest-nf:seq-leaf}
we have $\phi' = \emptyforest$, $\phi = \forest(A)$ and $A = Q_{\phi}$.
If $l = x \in \Names$ then we observe that conditions \ALLCOND hold for $\phi'$
under the assumption that they hold for $\phi$.
Therefore $\phi' = \forest(Q')$ with $Q' \congr Q_{\phi'}$,
and, by definition of $\forest$, $\phi = \forest(\new x.Q')$.
By condition \ref{lemma:forest-nf:name-uniq} we have $x \not\in X_{\phi'}$ so
$\new x.Q'
 \congr \new x. Q_{\phi'}
 \congr \new (X \union \set{x}).P_{\phi'}
 = Q_\phi$.

\end{proof}
  \section{Supplementary Material for Section~\ref{sec:t-compat}}

\subsection{Proof of \cref{th:hier-db}}

  First, it is immediate to see that every hierarchical term is depth-bounded.
Any \tcompat\ forest cannot repeat a type in a path,
which means that the number of base types in $\Types$ bounds the height
of \tcompat\ forests.
This automatically gives a bound on the depth of any \tcompat\ term.

We show the converse is not true
by presenting a depth-bounded process which is not hierarchical.
Take $P = (\bang{A} \parallel \bang{B} \parallel \bang{(C_1 + C_2)})$ where
\begin{align*}
  A   & = \tact.\new (a\tas \ty{a}).\out p<a> &
  C_1 & = \inp p(x).\bang{(\inp q(y).D)}     &
  D   & = \out x<y>                          \\
  B   & = \tact.\new (b\tas \ty{b}).\out q<b> &
  C_2 & = \inp q(x).\bang{(\inp p(y).D)}
\end{align*}
then $P$ is depth-bounded.
However we can show there is no choice
for \emph{consistent} annotations and $\Types$
that can prove it hierarchical.
Let $h$ be the height of $\Types$.
From $P$ we can reach, by reducing the $\tact$ actions of $A$ and $B$,
any of the terms
  $Q_{i,j} = P \parallel
    (\new a. \out p<a>)^i \parallel (\new b. \out q<b>)^j$
(omitting annotations) for $i,j \in \Nat$.
The choice for annotations can potentially assign a different type in $\Types$
to each $\restr a$ and $\restr b$.
Let $n, m \in \Nat$ be naturals strictly greater than $2 h$
and consider the reachable term $Q_{n, n m}$;
from this term we can reach a term
\[
  Q^{ab} =
    P \parallel
  \biggl(\new a.\Bigl(
    \bigl(\new b.
      D\subst{x -> a, y -> b}
    \bigr)^m \parallel
    \Bang{\inp q(y).D\subst{x -> a}}
    \Bigr)
  \biggr)^n
\]
by never selecting $C_2$ as part of a redex.
Each occurrence of $a$ and $b$ will have an annotation:
we assume type $t_a^i$ is assigned
to each occurrence $i \leq n$ of $\restr a$ in $Q^{ab}$
and a type $t_b^{i,j}$ is assigned to each occurrence $j$ of $\restr b$
under $\restr(a \tas t_a^i)$ in $Q^{ab}$.
Each occurrence of $\restr a$ in $Q^{ab}$ has in its scope
more than $h$ occurrences of $\restr b$.
We cannot extrude more than $h$ occurrences of $\restr b$
because we would necessarily violate \tcompat[ibility]
by obtaining a path of length greater than $h$
in the forest of the extruded term.
Therefore, w.l.o.g., we can assume that the types
  $t_b^{i,1},\dots,t_b^{i,h+1}$
are all descendants of $t_a^i$, for each $i \leq n$.
Pictorially, the parent relation in $\Types$ entails the relations
in \cref{fig:semantic-counterexample}
where the edges represent $\tlt_\Types$.

\begin{figure}[htb]
  \centering%
  \input{export/semantic-counterexample.patch.tex}%
  \caption{Structure of $\Types$ in the counterexample.}%
  \label{fig:semantic-counterexample}%
\end{figure}

The type associations of the restrictions in $Q^{ab}$
are already fixed in $Q_{n, nm}$.
From $Q_{n, nm}$ we can however also reach any of the terms
\[
  Q_b^i =
    P \parallel \cdots \parallel
  \biggl(\new (b\tas t_b^{i, 1}).\Bigl(
    \bigl(\new a.
      D\subst{x -> a, y -> b}
    \bigr)^n \parallel
    \Bang{\inp p(y).D\subst{x -> b}}
    \Bigr)
  \biggr)
\]
for $i \leq m$,
by making $C_2$ and $\new (b\tas t_b^{i,1}).\out q<b>$ react and then
repeatedly making $\bang{\inp q(y).D}$ react with each
$\new (a\tas t_a^j).\out p<a>$.
Let us consider $Q_b^1$.
As before, we cannot extrude more than $h$ occurrences of $a$
or we would break \tcompat[ibility].
We must however extrude $(a\tas t_a^1)$ to get \tcompat[ibility] since
$t_a^1 \tlt_\Types t_b^{1,1}$.
From these two facts we can infer that
there must be a type associated to one of the $a$, let it be $t_a^2$,
such that $t_a^1 \tlt_\Types t_b^{1,1} \tlt_\Types t_a^2$.
We can apply the same argument to $Q_b^2$ obtaining
$t_a^1 \tlt_\Types t_a^2 \tlt_\Types t_b^{2,1} \tlt_\Types t_a^3$.
Since $m > 2h$ we can repeat this $h+1$ times and get
$t_a^1 \tlt_\Types t_a^2 \tlt_\Types \dots \tlt_\Types t_a^{h+1}$
which contradicts the assumption that the height of $\Types$ is $h$.

The reason why the counterexample presented in the proof above
fails to be hierarchical is that (unboundedly many) names are used
in fundamentally different ways
in different branches of the execution.


  \section{Supplementary Material for Section~\ref{sec:typesys}}

\subsection{Proof of \cref{lemma:tied-tree}}

  We show that the claim holds in the case where $A_i$ is linked to $A_j$ in $P$.
From this, a simple induction over the length of linked-to steps required
to prove $i \tiedto{P} j$, can prove the lemma.

Suppose $i \linkedto{P} j$.
Let $Y = \freenames(A_i)
           \inters
         \freenames(A_j)
           \inters
         \set{x | (x \tas \type) \in X}$,
we have $Y \neq \emptyset$.
Both $A_i$ and $A_j$ are in the scope
of each of the restrictions bounding names $y \in Y$ in any
of the processes $Q$ in the congruence class of $P$,
hence, by definition of $\forest$, the nodes labelled with $A_i$ and $A_j$
generated by $\forest(Q)$ will have nodes labelled with $(y, \base(X(y)))$
as common ancestors.

\subsection{Some auxiliary lemmas}

\begin{lemma}
\label{lemma:tcompat-alpha}
  If\/ $\forest(P)$ is \tcompat\ then
  for any term $Q$ which is an \pre\alpha-renaming of $P$,
  $\forest(Q)$ is \tcompat.
\end{lemma}

\begin{proof}
  Straightforward from the fact that \tcompat[ibility]
  depends only on the type annotations.
\end{proof}

\begin{lemma}
\label{lemma:tcompat-takeout}
  Let $P = \new X.\Parallel_{i \in I} A_i$ be a \tcompat\ normal form,
  $Y \subseteq X$ and $J \subseteq I$.
  Then $P' = \new Y.\Parallel_{j \in J} A_j$ is \tcompat.
\end{lemma}

\begin{proof}
    Take a \tcompat\ forest $\phi \in \AST{P}$.
  By Lemma~\ref{lemma:tcompat-alpha} we can assume without loss of generality
  that $\phi = \forest(Q)$ where proving $Q \congr P$
  does not require \pre\alpha-renaming.
  Clearly, removing the leaves that do not correspond to sequential terms
  indexed by $Y$ does not affect the \tcompat[ibility] of $\phi$.
  Similarly, if a restriction $(x\tas\type)\in X$ is not in $Y$,
  we can remove the node of $\phi$ labelled with $(x, \base(\type))$
  by making its parent the new parent of its children.
  This operation is unambiguous under \mbox{\ref{nameuniq}}
  and does not affect \tcompat[ibility], by transitivity of $\tlt$.
  We then obtain a forest $\phi'$ which is \tcompat\ and that,
  by Lemma~\ref{lemma:forest-nf},
  is the forest of a term congruent to the desired normal form $P'$.

\end{proof}

  \section{Supplementary Material for Section~\ref{sec:soundness}}

\subsection{Some Elementary Properties of the Type System}

\begin{lemma}
\label{lemma:typesys-props}
  Let $P \in \PiNfAnnot$ and $\Env$, $\Env'$ be type environments.
  \begin{enumerate}[label=\alph*), ref={\ref{lemma:typesys-props}.\alph*}]
    \item if\/ $\Env \types P$ then $\freenames(P) \subseteq \domain(\Env)$;
      \label{lemma:typesys-domain}
    \item if $\domain(\Env') \inters \boundnames(P) = \emptyset$
          and $\freenames(P) \subseteq \domain(\Env)$,\\
          then $\Env \types P$ if and only if\/ $\Env \Env' \types P$;%
      \label{lemma:typesys-weakening}%
    \item if $P \congr P' \in \PiNfAnnot$ then, $\Env \types P$ if and only if\/ $\Env \types P'$.
      \label{lemma:typesys-congr}
  \end{enumerate}
\end{lemma}

\subsection{Proof of \Cref{lemma:phi-tcompat}}

Item~\ref{lemma:phi-tcompat:tcompat} is an easy induction on the cardinality of $X$.

Item~\ref{lemma:phi-tcompat:ast} requires more work.
By item~\ref{lemma:phi-tcompat:tcompat} $\Phi(P)$ is \tcompat\ so
$\Phi(P) \in \AST{P}$ proves that $P$ is \tcompat.

To prove the $\Leftarrow$-direction we assume that
$P = \new X.\Parallel_{i\in I} A_i$ is \tcompat{}
and proceed by induction on the cardinality of $X$
to show that $\Phi(P) \in \AST{P}$.
The base case is when $X = \emptyset$:
$ \Phi(P) = \Phi(\Parallel_{i\in I} A_i)
          = \Dunion_{i\in I} \set{A_i[]}
          = \forest(\Parallel_{i\in I} A_i)
          = \forest(P) \in \AST{P}
$.
For the induction step,
we observe that $X \neq \emptyset$ implies $\minrestr(X) \neq \emptyset$ so,
$Z \subset X$ and
for each $(x \tas \type) \in \minrestr(X)$, $Y_x \subset X$
since $x \not\in Y_x$.
This, together with Lemma~\ref{lemma:tcompat-takeout},
allows us to apply the induction hypotesis on the terms
  $P_x = \new Y_x.\Parallel_{j \in I_x} A_j$ and
  $P_R = \new Z.\Parallel_{r\in R} A_r$,
obtaining that there exist terms $Q_x \congr P_x$ and $Q_R \congr P_R$ such that
$\forest(Q_x) = \Phi(P_x)$ and $\forest(Q_R) = \Phi(P_R)$ where
all the forests $\forest(Q_x)$ and $\forest(Q_R)$ are \tcompat.
Let
$ Q = \Parallel \set{\new (x \tas \type). Q_x | (x \tas \type) \in \minrestr(X)}
      \parallel Q_R
$,
then $\forest(Q) = \Phi(P)$.
To prove the claim we only need to show that $Q \congr P$.
We have
$ Q \congr \Parallel
      \set{ \new (x \tas \type).\new Y_x.\Parallel_{j \in I_x} A_j
            | (x \tas \type) \in \minrestr(X) }
      \parallel P_R
$
and we want to apply extrusion to get
$ Q \congr
      \new Y_{\min}.
      \left( \Parallel_{i \in I_{\min}} A_i \right)
      \parallel P_R
$
for
  $I_{\min} = \Dunion \set{I_x | (x \tas \type) \in \minrestr(X)}$,
  $Y_{\min} = \minrestr(X) \dunion \Dunion \set{Y_x | (x \tas \type) \in \minrestr(X)}$
which adds an obligation to prove that
\begin{enumerate}[label=\roman*)]
  \item $I_x$ are all pairwise disjoint so that $I_{\min}$ is well-defined,
        \label{phi-tcompat:Ix-disjoint}
  \item $Y_x$ are all pairwise disjoint and all disjoint from $\minrestr(X)$
        so that $Y_{\min}$ is well-defined,
        \label{phi-tcompat:Yx-disjoint}
  \item $Y_x \inters \freenames(A_j) = \emptyset$ for every $j \in I_z$
        with $z \neq x$ so that we can apply the extrusion rule.
        \label{phi-tcompat:Yx-Iz}
\end{enumerate}

To prove condition~\ref{phi-tcompat:Ix-disjoint},
assume by contradiction that there exists an $i \in I$
and names $x, y \in \minrestr(X)$ with $x \neq y$,
such that both $x$ and $y$ are tied to $A_i$ in $P$.
By transitivity of the tied-to relation, we have $I_x = I_y$.
By Lemma~\ref{lemma:tied-tree} all the $A_j$ with $j \in I_x$
need to be in the same tree in any forest $\phi \in \AST{P}$.
Since $P$ is \tcompat\ there exist such a $\phi$ which is \tcompat\ and
has every $A_j$ as label of leaves of the same tree.
This tree will include a node $n_x$ labelled with $(x, \base(X(x)))$
and a node $n_y$ labelled with $(y, \base(X(y)))$.
By \tcompat[ibility] of $\phi$
and the existence of a path between $n_x$ and $n_y$
we infer $\base(X(x)) < \base(X(y))$ or $\base(X(y)) < \base(X(x))$
which contradicts the assumption that $x, y \in \minrestr(X)$.

Condition~\ref{phi-tcompat:Yx-disjoint} follows from condition~\ref{phi-tcompat:Ix-disjoint}:
suppose there exists a $(z \tas \type) \in X \inters Y_x \inters Y_y$ for $x \neq y$,
then we would have that $z \in \freenames(A_i) \inters \freenames(A_j)$
for some $i \in I_x$ and $j \in I_y$,
but then $i \tiedto{P} j$,
meaning that $i \in I_y$ and $j \in I_x$ violating condition~\ref{phi-tcompat:Ix-disjoint}.
The fact that $Y_x \inters \minrestr(X) = \emptyset$
follows from the definition of $Y_x$.
The same reasoning proves condition~\ref{phi-tcompat:Yx-Iz}.

Now we have
$ Q \congr
      \new Y_{\min}.
      \left( \Parallel_{i \in I_{\min}} A_i \right)
      \parallel \new Z.\Parallel_{r\in R} A_r $
and we want to apply extrusion again to get
$ Q \congr
      \new Y_{\min} Z. \Parallel \set{A_i | i \in (I_{\min} \dunion R)}
$
which is sound under the following conditions:
\begin{enumerate}[resume*]
  \item $Y_{\min} \inters Z = \emptyset$,
  \item $I_{\min} \inters R = \emptyset$,
  \item $Z \inters \freenames(A_i) = \emptyset$ for all $i \not\in R$
\end{enumerate}
of which the first two hold trivially by construction,
while the last follows from condition~\ref{phi-tcompat:X-complete} below,
as a name in the intersection of $Z$ and a $\freenames(A_i)$
would need to be in $X$ but not in $Y_{\min}$.
To be able to conclude that $Q \congr P$ it remains to prove that
\begin{enumerate}[resume*]
  \item $I = I_{\min} \dunion R$ and
        \label{phi-tcompat:I-complete}
  \item $X = Y_{\min} \dunion Z$
        \label{phi-tcompat:X-complete}
\end{enumerate}
which are also trivially valid by inspection of their definitions.
This concludes the proof for item~\ref{lemma:phi-tcompat:ast}.

Finally, for every $Q \in \PiAnnot$ such that
$Q \congr P$, $\Phi(P) \in \AST{Q}$ if and only if $\Phi(P) \in \AST{P}$
by definition of $\AST{-}$;
since $\Phi(P)$ is \tcompat\ we can infer that
$Q$ is \tcompat{} if and only if $\Phi(P) \in \AST{Q}$,
which proves item~\ref{lemma:phi-tcompat:congr}.
%

In light of \Cref{lemma:phi-tcompat},
we can turn the computation of
$\Phi_\Types(P)$ into an algorithm to check \tcompat[ibility] of $P$:
it is sufficient to compute $\Phi_\Types(P)$ and check at each step
that the sets $I_x$, $R$ form a partition of $I$ and
the sets $Y_x$, $Z$ form a partition of $X$.
If the checks fail $\Phi_\Types(P) \not\in \AST{P}$ and $P$ is not \tcompat,
otherwise the obtained forest is a witness of \tcompat[ibility].

\subsection{Further Properties of $\Phi_\Types(P)$}

\begin{lemma}
\label{lemma:phi-tied}
  Let $P = \new X . \Parallel_{i \in I} A_i \in \PiNfAnnot$ be a \tcompat\ normal form.
  Then for every trace
    $((x_1, t_1) \dots (x_k, t_k) \: A_j)$
  in the forest $\Phi(P)$,
  for every $i \in \set{1, \ldots, k}$, we have $x_i \ntiedto{P} j$ (i.e.
  $x_i$ is tied to $A_j$ in $P$).
\end{lemma}

\begin{proof}
  Straightforward from the definition of $I_x$ in $\Phi$:
  when a node labelled by $(x,t)$ is introduced,
  its subtree is extracted from a recursive call on a term
  that contains all and only the sequential terms that are tied to $x$.
\end{proof}

\begin{remark}
\label{remark:phi-nf}
  $\Phi(P)$ satisfies conditions
    \ref{lemma:forest-nf:seq-leaf},
    \ref{lemma:forest-nf:name-uniq} and
    \ref{lemma:forest-nf:scoping}
  of \cref{lemma:forest-nf}.
\end{remark}

\subsection{Proof of \Cref{lemma:subst}}

%
We prove the lemma by induction on the structure of $P$.
The base case is when $P \congr \zero$, where the claim trivially holds.

For the induction step,
let $P \congr \new X. \Parallel_{i \in I} A_i$
with $A_i = \Alt_{j \in J} \pi_{ij}.P_{ij}$,
for some finite sets of indexes $I$ and $J$.
Since the presence of replication does not affect the typing proof,
we can safely ignore that case as it follows the same argument.
Let us assume $\Env\types P$ and prove that $\Env\types P\subst{a->b}$.

Let $\Env'$ be $\Env \union X$.
From $\Env\types P$ we have
\begin{gather}
  \Env, X \types A_i
  \label{eq:subst-Ai}\\
      x \ntiedto{P} i
        \implies
          \base(\Env(\freenames(A_i))) \tlt \base(\type_x)
  \label{eq:subst-tied}
\end{gather}
for each $i \in I$ and $x \tas \type_x \in X$.
To extract from this assumptions a proof for $\Env\types P\subst{a->b}$,
we need to prove that \eqref{eq:subst-Ai} and \eqref{eq:subst-tied}
hold after the substitution.

Since the substitution does not apply to names in $X$
and the \emph{tied to} relation is only concerned
with names in $X$, the only relevant effect of the substitution is
modifying the set $\freenames(A_i)$ to
  $\freenames(A_i \subst{a->b}) =
      \freenames(A_i) \setminus \set{a} \union \set{b}$
when $a \in \freenames(A_i)$;
But since $\Env(a) = \Env(b)$ by hypothesis,
we have $\base(\Env(\freenames(A_i \subst{a->b}))) \tlt \base(\type_x)$.

It remains to prove \eqref{eq:subst-Ai} holds after the substitution as well.
This amounts to prove for each $j \in J$ that
$
    \Env' \types \pi_{ij}.P_{ij}
        \implies
    \Env' \types \pi_{ij}.P_{ij}\subst{a->b}
$;
we prove this by cases.

Suppose $\pi_{ij} = \out{\alpha}<\beta>$ for two names $\alpha$ and $\beta$,
then from $\Env' \types \pi_{ij}.P_{ij}$ we know the following
\begin{gather}
  \alpha \tas t_\alpha[\type_\beta] \in \Env'
  \qquad
  \beta \tas \type_\beta \in \Env'
  \label{eq:subst-out-types}\\
  \Env' \types P_{ij}
  \label{eq:subst-out-induction}
\end{gather}
Condition~\eqref{eq:subst-out-types} is preserved after the substitution
because it involves only types so, even if $\alpha$ or $\beta$ are $a$,
their types will be left untouched after they get substituted with $b$
from the hypothesis that $\Env(a) = \Env(b)$.
Condition~\eqref{eq:subst-out-induction} implies
$\Env' \types P_{ij}\subst{a->b}$ by inductive hypothesis.

Suppose now that $\pi_{ij} = \inp{\alpha}(x)$
and $P_{ij} \congr \new Y.\Parallel_{k \in K} A'_k$
for some finite set of indexes $K$;
by hypothesis we have:
\begin{gather}
  \alpha \tas t_{\alpha}[\type_x] \in \Env'
  \label{eq:subst-in-typea} \\
  \Env', x \tas \type_x \types P_{ij}
  \label{eq:subst-in-induction} \\
  \base(\type_x) \tlt t_{\alpha}
  \lor 
  \forall k \in K \st
    \migr{\pi_{ij}.P_{ij}}(k) \implies
      \base(\Env'(\freenames(A'_k)\setminus\set{\alpha})) \tlt t_{\alpha}
  \label{eq:subst-in-migrate}
\end{gather}
Now $x$ and $Y$ are bound names so they are not altered by substitutions.
The substitution $\subst{a->b}$ can therefore only be affecting
the truth of these conditions when
$\alpha =  a$ or when $a \in \freenames(A'_k)\setminus (Y \union \set{x})$.
Since we know $a$ and $b$ are assigned the same type by $\Env$
and $\Env \subseteq \Env'$,
condition~\eqref{eq:subst-in-typea} still holds when substituting $a$ for $b$.
Condition~\eqref{eq:subst-in-induction} holds by inductive hypotesis.
The first disjunct of condition~\eqref{eq:subst-in-migrate}
depends only on types, which are not changed by the substitution,
so it holds after applying it if and only if it holds before the application.
To see that the second disjunct also holds after the substitution
we observe that the \emph{migratable} condition depends on
$x$ and $\freenames(A'_k) \inters Y$
which are preserved by the substitution;
moreover, if $a \in \freenames(A'_k)\setminus\set{\alpha}$ then
$\Env'(\freenames(A'_k)\setminus\set{\alpha}) =
    \Env'(\freenames(A'_k\subst{a->b})\setminus\set{\alpha})$.

This shows that the premises needed to derive
$\Env', x \tas \type_x' \types \pi_{ij}.P_{ij}\subst{a->b}$
are implied by our hypothesis,
which completes the proof.


\subsection{Proof of \Cref{th:subj-red}}

   We will only prove the result for the case when $P \redto Q$ is caused by
a synchronising send and receive action
since the $\tact$ action case is similar and simpler.
From $P\redto Q$ we know that
$P \congr \new W.(S \parallel R \parallel C) \in \PiNfAnnot$ with
$S \congr (\out a<b>.\new Y_s.S')+M_s$ and
$R \congr (\inp a(x).\new Y_r.R')+M_r$
the synchronising sender and receiver respectively;
$Q \congr \new W Y_s Y_r.
  (S' \parallel R'\subst{x->b} \parallel C)$.
In what follows, let $W' = W Y_s Y_r$,
  $C = \Parallel_{h\in H} C_h$,
  $S' = \Parallel_{i\in I} S'_i$ and
  $R' = \Parallel_{j\in J} R'_j$,
all normal forms.

For annotated terms, the type system is syntax directed:
there can be only one proof derivation for each typable term.
By Lemma~\ref{lemma:typesys-congr}, from the hypothesis $\Env \types P$
we can deduce $\Env \types \new W.(S \parallel R \parallel C)$.
The proof derivation for this typing judgment
can only be of the following shape:
\begin{equation}
\label{eq:P-deriv}
  \mprset{sep=1em}
  \inferrule{
    \Env W \types S \\
    \Env W \types R \\
    \forall h \in H \st
      \Env W \types C_h \\
    \Premise
  }{
    \Env \types \new W.(S \parallel R \parallel C)
  }
\end{equation}
where $\Premise$ represents the rest of the conditions of the \ref{rule:conf} rule.%
\footnote{Note that $\Premise$ is trivially true by $P$-safety of $\Env$.}
The fact that $P$ is typable implies that each of these premises must be provable.
The derivation proving $\Env, W \types S$ must be of the form
\begin{equation}
\label{eq:S-deriv}
  \mprset{sep=1em}
  \inferrule{
    \inferrule*{
      a \tas t_a[\type_b] \in \Env W \\
      b \tas \type_b \in \Env W \\
      \Env W \types \new Y_s.S'
    }{
      \Env W \types \out a<b>.\new Y_s.S'
    }\\
    \Premise_{M_s}\\
  }{
    \Env \types \out a<b>.\new Y_s.S' + M_s
  }
\end{equation}
where $\Env W \types \new Y_s.S'$ is proved by an inference of the shape
\begin{equation}
\label{eq:Si-deriv}
  \inferrule{
    \forall i\in I\st
      \Env W Y_s \types S'_i \\
    \forall i\in I\st
      \Premise_{S'_i}
  }{
    \Env W \types \new Y_s.S'
  }
\end{equation}

Analogously, $\Env W \types R$ must be proved by an inference with the following shape
\begin{equation}
\label{eq:R-deriv}
  \mprset{sep=1em}
  \inferrule{
    \inferrule*{
      a \tas t_a[\type_x] \in \Env W \\
      \Env W, x \tas \type_x \types \new Y_r.R'\\
      \Premise_{R'}
    }{
      \Env W \types \inp a(x).\new Y_r.R'
    }\\
    \Premise_{M_r}
  }{
    \Env W \types \inp a(x).\new Y_r.R' + M_r
  }
\end{equation}
and to prove $\Env W, x \tas \type_x \types \new Y_r.R'$
\begin{equation}
\label{eq:Rj-deriv}
  \inferrule{
    \forall j\in J\st
      \Env W, x \tas \type_x, Y_r \types R'_j\\
    \forall j\in J\st
      \Premise_{R'_j}
  }{
    \Env W, x \tas \type_x \types \new Y_r.R'
  }
\end{equation}

We have to show that from this hypothesis we can infer that $\Env \types Q$
or, equivalently (by Lemma~\ref{lemma:typesys-congr}), that $\Env \types Q'$
where $Q' = \new W Y_s Y_r. (S' \parallel R'\subst{x->b} \parallel C)$.
The derivation of this judgment can only end with an application of~\ref{rule:conf}:
\[\mprset{sep=1em}
  \inferrule{
    \forall i \in I \st
      \Env W' \types S'_i \\
    \forall j \in J \st
      \Env W' \types R'_j\subst{x->b} \\
    \forall h \in H \st
      \Env W' \types C_h \\
    \Premise'
  }{
    \Env \types \new W'. (S' \parallel R'\subst{x->b} \parallel C)
  }
\]
In what follows we show how we can infer these premises are provable
as a consequence of the provability of the premises of the proof of
$\Env \types \new W.(S \parallel R \parallel C)$.

From Lemma~\ref{lemma:typesys-weakening} and \ref{nameuniq},
$\Env W Y_s \types S'_i$ from~\eqref{eq:Si-deriv} implies
$\Env W' \types S'_i$ for each $i \in I$.

Let $\Env_r = \Env W, x \tas \type_x$.
We observe that by \eqref{eq:S-deriv} and \eqref{eq:R-deriv}, $\type_x = \type_b$.
From \eqref{eq:R-deriv} we know that $\Env_r Y_r \types R'_j$
which, by Lemma~\ref{lemma:subst}, implies $\Env_r Y_r \types R'_j\subst{x->b}$.
By Lemma~\ref{lemma:typesys-weakening} we can infer
$\Env_r Y_r Y_s \types R'_j\subst{x->b}$ and by applying the same lemma again
using $\freenames(R'_j\subst{x->b}) \subseteq \domain(\Env W Y_r Y_s)$
and \ref{nameuniq} we obtain
$\Env W' \types R'_j\subst{x->b}$.

Again applying Lemma~\ref{lemma:typesys-weakening} and \ref{nameuniq},
we have that $\Env W \types C_h$ implies $\Env W' \types C_h$ for each $h \in H$.

To complete the proof we only need to prove that for each
$A \in \set{S'_i | i \in I} \union \set{R'_j | j \in J} \union \set{C_h | h \in H}$,
$
  \Premise' = \forall (x \tas \type_x) \in W'\st
    x \text{ tied to } A \text{ in } Q' \implies
      \base(\Env(\freenames(A))) < \base(\type_x)
$
holds.
This is trivially true by the hypothesis that $\Env$ is $P$-safe.

\subsection{Proof of \Cref{th:typed-tshaped}}

   We will consider the input output synchronisation case as the $\tau$ action one is similar and simpler.
We will further assume that the sending action $\out a<b>$ is such that
$\restr (a \tas \type_a)$ and $\restr (b \tas \type_b)$ are both active restrictions of $P$,
i.e.~$(a \tas \type_a) \in W$, $(b \tas \type_b) \in W$
     with $P \congr \new W.(S \parallel R \parallel C)$.
The case when any of these two names is a free name of $P$ can be easily handled
with the aid of the assumption that $\Env$ is $P$-safe.

As in the proof of Theorem~\ref{th:subj-red}, the derivation of $\Env \types P$
must follow the shape of \eqref{eq:P-deriv}.

From \tshaped[ness] of $P$ we can conclude that
both $\new Y_s.S'$ and $\new Y_r.R'$ are \tshaped.
We note that substitutions do not affect \tcompat[ibility]
since they do not alter the set of bound names and their type annotations.
Therefore, we can infer that $\new Y_r.R'\subst{a->b}$ is \tshaped.
By Lemma~\ref{lemma:phi-tcompat} we know that
  $\phi   = \Phi(\new W.(S \parallel R \parallel C)) \in \AST{P}$,
  $\phi_r = \Phi(\new Y_r.R'\subst{a->b}) \in \AST{\new Y_r.R'\subst{x->b}}$ and
  $\phi_s = \Phi(\new Y_s.S') \in \AST{\new Y_s.S'}$.
Let $\phi_r = \phimig \dunion \phinonmig$
where only $\phimig$ contains a leaf labelled with a term with $b$ as a free name.
These leaves will correspond to the continuations $R'_j$
that migrate in $\inp a(x).\new Y_r.R'$,
after the application of the substitution $\subst{x->b}$.
By assumption, inside $P$ both $S$ and $R$ are in the scope of the restriction bounding $a$
and $S$ must also be in the scope of the restriction bounding $b$.
Let $t_a = \base(\type_a)$ and $t_b = \base(\type_b)$,
$\phi$ will contain two leaves $n_S$ and $n_R$
labelled with $S$ and $R$ respectively,
having a common ancestor $n_a$ labelled with $(a, t_a)$;
$n_S$ will have an ancestor $n_b$ labelled with $(b, t_b)$.
Let $p_a$, $p_S$ and $p_R$ be the paths in $\phi$ leading from a root to $n_a$, $n_S$ and $n_R$ respectively.
By \tcompat[ibility] of $\phi$, we are left with only two possible cases:
either \begin{inparaenum}[label=\arabic*)]
\item $t_a \tlt t_b$\label{case:ab} or \item $t_b \tlt t_a$.\label{case:ba}
\end{inparaenum}

Let us consider case~\ref{case:ab} first.
The tree in $\phi$ to which the nodes $n_S$ and $n_R$ belong,
would have the following shape:
\begin{center}%
  \inputfigwrap[y=.6cm,x=.6cm,baseline=0]{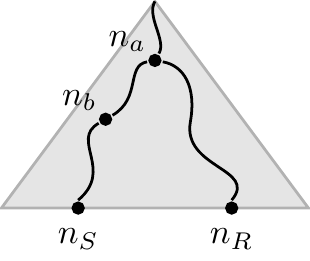}
\end{center}
\newcommand{\parentins}{\parent_{\mathrm{ins}}}
Now, we want to transform $\phi$, by manipulating this tree,
into a forest $\phi'$ that is
\tcompat\ by construction and such that
there exists a term $Q' \congr Q$ with $\forest(Q') = \phi'$,
so that we can conclude $Q$ is \tshaped.

To do so, we introduce the following function,
taking a labelled forest $\phi$, a path $p$ in $\phi$ and a labelled forest $\rho$
and returning a labelled forest:
\[
  \treeins(\phi, p, \rho) \is (
    N_\phi \dunion N_\rho,
    \parent_\phi \dunion \parent_\rho \dunion \parentins,
    \ell_\phi \dunion \ell_\rho
  )
\]
where $n\parentins n'$ if
  $n' \in \min_{\parent_\rho}(N_\rho)$
and if $\ell_\rho(n') = (y, t_y)$ then
\[
  n \in \max_{\parent_\phi}\set{m \in p | \ell_\phi(m) = (x,t_x), t_x < t_y}
\]
or if $\ell_\rho(n') = A$ then
\[
  n \in \max_{\parent_\phi}\set{m \in p | \ell_\phi(m) = (x,t_x), x \in \freenames(A)}.
\]
Note that for each $n'$, since $p$ is a path, there can be at most one $n$ such that $n \parentins n'$.

To obtain the desired $\phi'$,
we first need to remove the leaves $n_S$ and $n_R$ from $\phi$,
as they represent the sequential processes which reacted,
obtaining a forest $\phi_C$.
We argue that the $\phi'$ we need is indeed
\begin{align*}
  \phi' &= \treeins(\phi_1, p_S, \phimig)    \\
  \phi_1 &= \treeins(\phi_2, p_R, \phinonmig)\\
  \phi_2 &= \treeins(\phi_C, p_S, \phi_s)
\end{align*}
It is easy to see that, by definition of $\treeins$, $\phi'$ is \tcompat:
$\phi_C$, $\phi_s$, $\phinonmig$ and $\phimig$ are \tcompat\ by hypothesis,
$\treeins$ adds parent-edges only when they do not break \tcompat[ibility].

To prove the claim we need to show that $\phi'$ is the forest of a term congruent to
  $\new W Y_s Y_r.(S' \parallel R'\subst{x->b} \parallel C)$.
\newcommand{\Jnonmig}{J_{\neg\mathrm{mig}}}
\newcommand{\Jmig}{J_{\mathrm{mig}}}
Let $R' = \Parallel_{j \in J} R'_j$,
    $\Jmig = \set{j \in J | x \ntiedto{\new Y_r.R'} j}$,
    $\Jnonmig = J \setminus \Jmig$
and $Y'_r = \set{(x \tas \type) \in Y_r
                  | x \in \freenames(R'_j), j \in \Jnonmig}$.
We know that no $R'_j$ with $j \in \Jnonmig$ can contain $x$ as a free name
so $R'_j\subst{x->b} = R'_j$.
Now suppose we are able to prove that conditions
  \ref{lemma:forest-nf:seq-leaf},
  \ref{lemma:forest-nf:name-uniq} and
  \ref{lemma:forest-nf:scoping}
of Lemma~\ref{lemma:forest-nf} hold for
  $\phi_C$, $\phi_1$, $\phi_2$ and $\phi'$.
Then we could use Lemma~\ref{lemma:forest-nf} to prove
\begin{enumerate}[label=\alph*)]
  \item $\phi_C = \forest(Q_C)$, $Q_C \congr Q_{\phi_C} = \new W.C$,
  \item $\phi_2 = \forest(Q_2)$, $Q_2 \congr Q_{\phi_2} = \new W Y_s.(S' \parallel C)$,
  \item $\phi_1 = \forest(Q_1)$, $Q_1 \congr Q_{\phi_1} = \new W Y_s Y'_r.(S' \parallel \Parallel_{j \in \Jnonmig} R'_j \parallel C)$,
  \item $\phi' = \forest(Q')$, $Q' \congr Q_{\phi'} = \new W Y_s Y_r.(S' \parallel R'\subst{x->b} \parallel C) \congr Q$
\end{enumerate}
(it is straightforward to check that $\phi_C, \phi_2, \phi_1$ and $\phi'$ have the right sets of nodes and labels to give rise to the right terms).
We then proceed to check for each of the forests above that they satisfy
conditions
  \ref{lemma:forest-nf:seq-leaf},
  \ref{lemma:forest-nf:name-uniq} and
  \ref{lemma:forest-nf:scoping},
thus proving the theorem.

Condition~\ref{lemma:forest-nf:seq-leaf} requires that
only leafs are labelled with sequential processes,
condition that is easily satisfied by all of the above forests
since none of the operations involved in their definition alters this property
and the forests $\phi$, $\phi_s$ and $\phi_r$ satisfy it by construction.

Similarly, since $\new W.(S \parallel R \parallel C)$ is a normal form
it satisfies \ref{nameuniq}, \ref{lemma:forest-nf:name-uniq} is satisfied
as we never use the same name more than once.

Condition \ref{lemma:forest-nf:scoping} holds on $\phi$ and hence it holds on $\phi_C$
since the latter contains all the nodes of $\phi$ labelled with names.

Now consider $\phi_s$: in the proof of Theorem~\ref{th:subj-red} we established
that $\Env \types P$ implies that the premises $\Premise_{S'_i}$
from \eqref{eq:Si-deriv} hold, that is
$\base(\Env W(\freenames(S'_i))) \tlt \base(\type_x)$
holds for all $S'_i$ for $i \in I$ and all $(x \tas \type_x) \in Y_s$
such that $x \ntiedto{\new Y_s.S'} i$.
Since $\freenames(S'_i) \inters W \subseteq \freenames(S')$ we know that
every name $(w \tas \type_w) \in W$ such that $w \in \freenames(S'_i)$
will appear as a label $(w, \base(\type_w))$ of a node $n_w$ in $p_S$.
Therefore, by definition of $\treeins$,
we have that for each $n \in N_{\phi_C}$, $n_w <_{\phi_2} n$;
in other words, in $\phi_2$, every leaf in $N_{\phi_s}$ labelled with $S'_i$
is a descendent of a node labelled with $(w, \base(\type_w))$
for each $(w \tas \type_w) \in W$ with $w \in \freenames(S'_i)$.
This verifies condition \ref{lemma:forest-nf:scoping} on $\phi_2$.

Similarly, by \eqref{eq:Rj-deriv} the following premise must hold:
$ \base(\Env W(\freenames(R'_j))) \tlt \base(\type_x) $
for all $R'_j$ for $j \in J$ and all $(y \tas \type_y) \in Y_r$
such that $y \ntiedto{\new Y_r.R'} j$.
We can then apply the same argument we applied to $\phi_2$ to show
that condition \ref{lemma:forest-nf:scoping} holds on $\phi_1$.

From \eqref{eq:R-deriv} and the assumption $t_a \tlt t_b$,
we can conclude that the following premise must hold:
$\base(\Env W(\freenames(R'_j) \setminus \set{a})) \tlt t_a$
for each $j \in J$ such that $R'_j$ is migratable in $\inp a(x).\new Y_r.R'$,
i.e~$j \in \Jmig$.
From this we can conclude that for every name $(w \tas \type_w) \in W$
such that $w \in \freenames(R'_j\subst{x->b})$ with $j \in \Jmig$
there must be a node in $p_a$ (and hence in $p_S$) labelled with $(w, \base(\type_w))$.
Now, some of the leaves in $\phimig$ will be labelled with terms having $b$ as a free name;
we show that in fact every node in $\phimig$ labelled with a
$(y, t_y)$ is indeed such that $t_y \tlt t_b$.
From the proof of Theorem~\ref{th:subj-red} and Lemma~\ref{lemma:subst}
we know that from the hypothesis we can infer that
$\Env W \types \new Y_r.R'\subst{x->b}$ and hence that
for each $j \in \Jmig$ and each $(y \tas \type_y) \in Y_r$,
if $y$ is tied to $R'_j\subst{x->b}$ in $\new Y_r.R'\subst{x->b}$ then
$\base(\Env W(R'_j\subst{x->b})) \tlt \base(\type_y)$.
By Lemma~\ref{lemma:phi-tied} we know that every root of $\phimig$
is labelled with a name $(y, t_y)$ which is tied to each of the leaves in its tree.
Therefore each such $t_y$ satisfies $\base(\Env W(R'_j\subst{x->b})) < t_y$.
By construction, there exists at least one $j \in \Jmig$ such that
$x \in \freenames(R'_j)$ and consequently such that
$b \in \freenames(R'_j\subst{x->b})$.
From this and $b \in W$ we can conclude $t_b < t_y$ for $t_y$ labelling a root in $\phimig$.
We can then conclude that
$\set{n_b} = \max_{\parent_{\phi_2}}\set{m \in p_S | \ell_\phi(m) = (z,t_z), t_z < t_y}$
for each $t_y$ labelling a root of $\phimig$,
which means that each tree of $\phimig$ is placed as a subtree of $n_b$ in $\phi'$.
This verifies condition~\ref{lemma:forest-nf:scoping} for $\phi'$
completing the proof.

Pictorially, the tree containing $n_S$ and $n_R$ in $\phi$ is now transformed
in the following tree in $\phi'$:
\begin{center}
  \inputfigwrap[scale=.65,baseline=0]{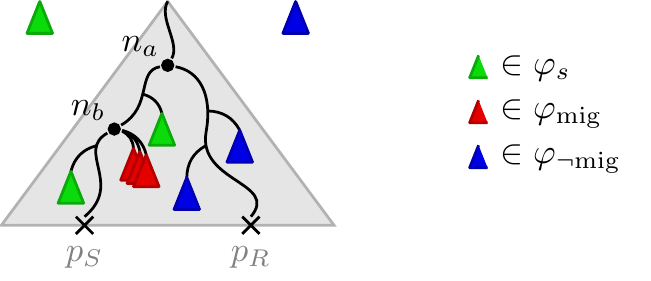}
\end{center}

Case~\ref{case:ba} --- where $t_b < t_a$ --- is simpler as the migrating
continuations can be treated just as the non-migrating ones.
%

\subsection{Role of $\phimig$, $\phinonmig$ and $\treeins$}

To illustrate the role of $\phimig$, $\phinonmig$
and the $\treeins$ operation in the above proof,
we show an example that would not be typable if we choose a simpler
``migration'' transformation.

Consider the normal form $P = \new a\:b\:c.(\bang{A} \parallel \out a<c>)$
where
  $A = \inp a(x).\new d.(\out a<d> \parallel \out b<x>)$.
To make types consistent we need annotations satisfying
$a \tas t_a[t]$, $b \tas t_b[t]$, $c \tas t$ and $d \tas t$.
Any $\Types$ satisfying the constraints $t_b \tlt t_a \tlt t$ would allow us
to prove $\emptyset\types P$;
let then $\Types$ be the forest with $\ty{b} \parent \ty{a} \parent \ty{t}$
with $t_a = \ty{a}$, $t_b = \ty{b}$ and $t = \ty{t}$.
Let $P' = \new a\:b\:c\:d.(\bang{A} \parallel \out a<d> \parallel \out b<c>)$
be the (only) successor of $P$.
The following picture shows $\Phi(P)$ in the middle,
on the left a forest in $\AST{P'}$ extracted by just putting
the continuation of $A$ under the message,
on the right the forest obtained by using $\treeins$
on the non-migrating continuations of $A$:
\begin{center}
  \inputfigwrap[yscale=.7]{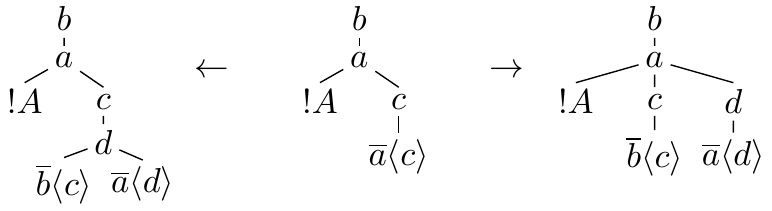}
\end{center}
Clearly, the tree on the left is not \tcompat\ since $c$ and $d$
have the same base type $t$.
Instead, the tree on the right can be obtained because $\treeins$
inserts the non-migrating continuation as close to the root as possible.

  \section{Supplementary Material for Section~\ref{sec:inference}}

\subsection{A type inference example}

Take the term
  $\Bang{\tact.\new s\:c.P}$
of \cref{ex:server-client}.
We start by annotating each restriction $\restr a$ with a fresh type variable
  $\restr (a\tas \typevar_a)$.
Then we perform a type derivation as in \cref{ex:servers-typing},
obtaining the following data-flow constraints:
\begin{align*}
  \typevar_s &= t_s[\typevar_z] &
  \typevar_c &= t_c[\typevar_x] &
  \typevar_x &= t_x[\typevar_y] &
  \typevar_z &= \typevar_x = \typevar_m = t_z[\typevar_d]
\end{align*}
from which we learn that:
\begin{itemize}
  \item $\typevar_d$ is unconstrained;
    we use the base type variable $t_d$ for $\base(\typevar_d)$;
  \item $\typevar_y = \typevar_d$;
  \item $t_x = t_z$ and $\base(\typevar_m) = t_x$.
\end{itemize}
We can therefore completely specify the types just by associating
$t_s, t_c, t_x$ and $t_d$ to nodes in a forest:
all the types would be determined as a consequence
of the data-flow constraints,
apart from $\typevar_d$ to which we can safely assign the type $t_d$.

During the type derivation we also collected
the following base type constraints:
\begin{gather*}
  \base(\typevar_z) \tlt \base(\typevar_d)  \qquad
  \base(\typevar_c) \tlt \base(\typevar_m)  \qquad
  \base(\typevar_c) \tlt \base(\typevar_x)  \\
  \base(\typevar_x) \tlt \base(\typevar_c)  \lor
  \base(\typevar_s) \tlt \base(\typevar_c)
\end{gather*}
These can be simplified and normalised using the equations on types seen above
obtaining the set
\[
  \C{\new s\:c.P} =
    \set{t_x \tlt t_c \lor t_s \tlt t_c,\,
         t_c \tlt t_x,\,
         t_x \tlt t_d}
\]
Hence any choice of $\Types \supseteq \set{t_x, t_c, t_s, t_d}$
such that $t_s \tlt_\Types t_c \tlt_\Types t_x \tlt_\Types t_d$
would make the typing succeed.

  \section{Supplementary Material for Section~\ref{sec:verification}}

\subsection{Encoding of Reset nets}

A \emph{reset net} $N$ with $n$ places
is a finite set of transitions of the form
  $(\vect{u},R)$
where $\vect{u}\in \set{-1,0,+1}^n$ is the update vector and 
      $R \subseteq \set{1,\dots,n}$ is the reset set.
A \emph{marking} $\vect{m}$ is a vector in $\Nat^n$;
a transition $(\vect{u},R)$ is said to be enabled at $\vect{m}$
if $\vect{m} - \vect{u} > \vect{0}$.
The semantics of a reset net $N$ with initial marking $\vect{m}_0$
is the transition system $(\Nat^n, \pnstep, \vect{m}_0)$
where $\vect{m} \pnstep \vect{m}'$ if there exists a transition
$(\vect{u},R)$ in $N$ that
  is enabled in $\vect{m}$ and
  such that \[
    \vect{m}'_i =
    \begin{cases}
      \vect{m}_i + \vect{u}_i \CASE i \not\in R \\
      0                       \CASE i \in R
    \end{cases}
  \]

To simulate place $i$ in a reset net we can construct a term that implements a
counter with increment and reset:
\[
  C_i = \Bang{
    \inp p_i(t).\bigl(
        \inpz \ch{inc}_i.(\outz t \parallel \out p_i<t>)
      + \inpz \ch{dec}_i.(\inpz t.\out p_i<t>)
      + \inpz \ch{rst}_i.(\new t'_i.\out p_i<t'_i>)
    \bigr)
  }
\]
Here, the number of processes $\outz t$ in parallel with $\out p_i<t>$
represent the current value of the marking in place $i$.
A transition $(\vect{u}, R)$ is encoded as a process
$
  T_{\vect{u}, R} = \Bang{
    \inpz{\ch{valid}}.D_{\vect{u}}.I_{\vect{u}}.Z_R.\outz{\ch{valid}}
  }
$
where
  $D_{\vect{u}} = \outz {\ch{dec}_{j_1}}.\cdots.\outz {\ch{dec}_{j_k}}$ with
  $\set{j_1, \dots, j_k} = \set{j | \vect{u}_j < 0}$,
  $I_{\vect{u}} = \outz {\ch{inc}_{i_1}}.\cdots.\outz {\ch{inc}_{i_l}}$ with
  $\set{i_1, \dots, i_l} = \set{i | \vect{u}_i > 0}$, and
  $Z_R = \outz {\ch{rst}_{r_1}}.\cdots.\outz {\ch{rst}_{r_m}}$ with
  $R = \set{r_1, \dots, r_l}$.

A marking $\vect{m}$ is encoded by a process
\[
  P_{N,\vect{m}} = \outz{\ch{valid}}
    \parallel
  \Parallel_{1 \leq i \leq n} \left(
      \out p_i<t_i> \parallel
      C_i \parallel
      \outz{t_i}^{\vect{m}_i}
    \right)
    \parallel
  \Parallel_{(\vect{u}, R) \in N} T_{\vect{u}, R}
\]
Actions on the name $\ch{valid}$ act as a global lock:
a transition may need many steps to complete, but by acquiring and releasing
$\ch{valid}$ it can ensure no other transition will fire in between.
If a transition tries to decrement a counter below zero, the counter would deadlock causing $\outz{\ch{valid}}$ to be never released again.
Therefore, the encoding preserves coverability:
  $\vect{m}$ is coverable in $N$ from $\vect{m}_0$
    if and only if
  $P_{N,\vect{m}}$ is coverable from $P_{N,\vect{m}_0}$.
Reachability is not preserved because each reset would generate some `garbage'
term $\new t.(\outz t \parallel \ldots \parallel \outz t)$
and thus, even when $\vect{m}$ is reachable,
$P_{N,\vect{m}}$ might not be reachable alone,
but only in parallel with some garbage.

The reader can verify that any encoding $P$ can be typed under the hierarchy
\begin{multline*}
  \ch{valid} <
    \ch{inc}_1 < \ch{dec}_1 < \ch{rst}_1 < t_1 < p_1 <
    \cdots <\\
    \ch{inc}_n < \ch{dec}_n < \ch{rst}_n < t_n < p_n <
    t'_1 < \cdots < t'_n
\end{multline*}
by annotating each restriction $\restr t'_i$ as
  $\restr (t'_i\tas t'_i)$
and using the \pre P-safe environment
$\set{(x\tas x) | x \in \freenames(P)}$.

\subsection{A weak encoding of Minsky machines}

\newcommand{\INC}{\mathtt{inc}}
\newcommand{\DEC}{\mathtt{dec}}
A $k$-counters Minsky machine is a finite list of instructions $\lstc{I}{n}$
each of which can be either an increase or a decrease command.
An increase command $\INC\ i\ j$
increases counter $i$ and jumps to instruction $I_j$.
A decrease command $\DEC\ i\ j_1\ j_2$
decreases counter $i$ jumping to instruction $I_{j_1}$
if the counter is greater than zero,
or jumps to $I_{j_2}$ otherwise.
We implement a counter $i$ with the process $C_i$ of \cref{ex:counter}.
An increase $I_m = \INC\ i\ j$ is encoded by
$\Bang{\inpz{i_m}.\outz{\ch{inc}_i}.\outz{i_j}}$.
A decrease $I_m = \DEC\ i\ j_1\ j_2$ is encoded by
$\Bang{\inpz{i_m}.(
  \outz{\ch{dec}_i}.\outz{i_{j_1}} + \outz{\ch{rst}_i}.\outz{i_{j_2}}
)}$.
A configuration of a Minsky machine is the vector of values of its registers $\lstc{r}{k}$ and the current instruction $j$;
its encoding is the term
\[
  \Parallel_{1\leq i \leq k}
      \new t_i.( \out p_i<t_i> \parallel \outz{t_i}^{r_i} )
    \parallel
  \outz{i_j}
    \parallel
  \Parallel_{1\leq m \leq n} P_{I_m}
\]
where $P_{I_m}$ is the encoding of the instruction $I_m$.

When a counter is zero, performing a decrease command on it in the encoding
presents a non-deterministic choice between sending a decrease or a reset signal to the counter.
In the branch where the decrease signal is sent,
the counter process will deadlock,
ending up in a term that is clearly not an encoding
of a configuration of the Minsky machine.
If instead a reset signal is sent, the counter will refresh the name $t$ with a new name, but the old one would be discarded as there is no sequential term which knows it.

When a counter is not zero, the branch where the decrease signal is sent will simply succeed, while the resetting one will generate some `garbage' term $\new t.(\outz t \parallel \ldots \parallel \outz t)$ in parallel with the rest of the encoding of the Minsky machine's configuration.

A configuration of the machine is thus reachable if and only if
its encoding (without garbage) is reachable from the encoding of the machine.
This proves \cref{th:reach-undec}.

\fi

\end{document}